\documentclass[nonacm]{acmart}
\usepackage{thm-restate}
\usepackage{cleveref} 
\usepackage[T1]{fontenc}
\usepackage{textcomp}
\usepackage[USenglish]{babel}
\usepackage{amsmath,amsthm}
\usepackage{color,enumerate,array,graphicx,tabularx,listings}
\usepackage[left,mathlines]{lineno}
\usepackage{algorithm,float,algorithmicx}
\usepackage[noend]{algpseudocode}
\usepackage{rotating, multirow, makecell,adjustbox}
\usepackage[normalem]{ulem}
\usepackage{dsfont}
\usepackage{forest}
\usepackage{tikz}
\usetikzlibrary{graphs}
\usetikzlibrary{graphs.standard} 
\usepackage{hyperref}

\usepackage{xspace}

\usepackage{float}
\newtheorem{theorem}{Theorem}
\newtheorem{lemma}[theorem]{Lemma}

\newtheorem{definition}{Definition}

\newtheorem{claim}[theorem]{Claim}
\newtheorem{invariant}[theorem]{Invariant}

\newcommand{\m}[1]{\ensuremath{\mathcal{#1}}}
\algnewcommand\algorithmiccase{\textbf{case}}
\algnewcommand\StateCase[1]{\State\hphantom{x}\ #1 \algorithmicthen} 
\algdef{SE}[CASE]{Case}{EndCase}[1]{\algorithmiccase\ #1\ \algorithmicthen}{\algorithmicend\ \algorithmiccase}%
\algdef{SE}{Begin}{End}{\textbf{begin}}{\textbf{end}}
\algnewcommand{\IfL}[1]{\State\algorithmicif\ #1\ \algorithmicthen}
\algnewcommand{\EndIfL}{\unskip\ \algorithmicend\ \algorithmicif}
\newcommand{\rd}[1]{\textsc{read}}
\newcommand{\wrt}[1]{\textsc{write}}
\newcommand{\adt}[1]{\textsc{audit#1}}

\newcommand{\mwrt}[1]{\textsc{maxWrite}}
\newcommand{\wrtm}{\textsc{writeMax}}
\newcommand{\upd}[1]{\textsc{update#1}}

\newcommand{\scn}[1]{\textsc{scan#1}}

\newcommand{\cas}{\ensuremath{\mathsf{compare\&swap}}}

\newcommand{\fx}{\ensuremath{\mathsf{fetch\&xor}}}

\newcommand{\ind}[3]{#1\stackrel{#2}{\sim}#3}
\newcommand{\indp}[2]{\ind{#1}{p}{#2}}

\newcommand*\patchAmsMathEnvironmentForLineno[1]{%
  \expandafter\let\csname old#1\expandafter\endcsname\csname #1\endcsname
  \expandafter\let\csname oldend#1\expandafter\endcsname\csname end#1\endcsname
  \renewenvironment{#1}%
     {\linenomath\csname old#1\endcsname}%
     {\csname oldend#1\endcsname\endlinenomath}}%
\newcommand*\patchBothAmsMathEnvironmentsForLineno[1]{%
  \patchAmsMathEnvironmentForLineno{#1}%
  \patchAmsMathEnvironmentForLineno{#1*}}%
\AtBeginDocument{%
  \patchBothAmsMathEnvironmentsForLineno{equation}%
  \patchBothAmsMathEnvironmentsForLineno{align}%
  \patchBothAmsMathEnvironmentsForLineno{flalign}%
  \patchBothAmsMathEnvironmentsForLineno{alignat}%
  \patchBothAmsMathEnvironmentsForLineno{gather}%
  \patchBothAmsMathEnvironmentsForLineno{multline}}

\newif\ifannote
\annotetrue  

\ifannote
    
    \newcommand{\anncomment}[3]{{\color{#1}[#2: #3]}}
\else
    
    \newcommand{\anncomment}[3]{}
\fi

\newcommand{\af}[1]{{#1}}

\newcommand{\ha}[1]{{#1}}
  
\title{Auditing without Leaks Despite Curiosity}
\author{Hagit Attiya}
\affiliation{%
  \institution{Technion}
  \city{Haifa}
  \country{Israel}
  }
\email{hagit@cs.technion.ac.il}

\author{Antonio Fernández Anta}
\affiliation{%
  \institution{IMDEA Software \& Networks Inst.}
  \city{Madrid}
  \country{Spain}
}
\email{antonio.fernandez@imdea.org}

\author{Alessia Milani}
\affiliation{%
  \institution{Aix Marseille Univ, CNRS, LIS}
  \city{Marseille}
  \country{France}
}
\email{alessia.milani@lis-lab.fr}

\author{Alexandre Rapetti}
\affiliation{%
  \institution{Université Paris-Saclay, CEA, List}
  \city{Palaiseau}
  \country{France}
}
\email{alexandre.rapetti@cea.fr}

\author{Corentin Travers}
\affiliation{%
  \institution{Aix Marseille Univ, CNRS, LIS}
  \city{Marseille}
  \country{France}
}
\email{corentin.travers@lis-lab.fr}

\acmYear{2025}\copyrightyear{2025}
\setcopyright{rightsretained}
\acmConference[PODC '25]{ACM Symposium on Principles of Distributed Computing}{June 16--20, 2025}{Huatulco, Mexico}
\acmBooktitle{ACM Symposium on Principles of Distributed Computing (PODC '25), June 16--20, 2025, Huatulco, Mexico}
\acmDOI{10.1145/3732772.3733516}
\acmISBN{979-8-4007-1885-4/25/06}

\begin{document}

\begin{abstract}
\textit{Auditing} data accesses helps preserve privacy and ensures accountability by allowing one to determine who accessed (potentially sensitive) information. A prior formal definition of register auditability was based on the values returned by read operations, \emph{without accounting for cases where a reader might learn a value without explicitly reading it or gain knowledge of data access without being an auditor}.

This paper introduces a refined definition of auditability that focuses on when a read operation is \emph{effective}, rather than relying on its completion and return of a value. Furthermore, we formally specify the constraints that \textit{prevent readers from learning values they did not explicitly read or from auditing other readers' accesses.}

Our primary algorithmic contribution is a wait-free implementation of a \emph{multi-writer, multi-reader register} that tracks effective reads while preventing unauthorized audits. The key challenge is ensuring that a read is auditable as soon as it becomes effective, which we achieve by combining value access and access logging into a single atomic operation. Another challenge is recording accesses without exposing them to readers, which we address using a simple encryption technique (one-time pad).

We extend this implementation to an \emph{auditable max register} 
that tracks the largest value ever written. 
The implementation deals with the additional challenge posed by the
max register semantics, which allows readers to learn prior values without reading them.

The max register, in turn, serves as the foundation for implementing an \emph{auditable snapshot} object and, more generally, \emph{versioned types}. These extensions maintain the strengthened notion of auditability, appropriately adapted from multi-writer, multi-reader registers.
\end{abstract}

\ccsdesc[500]{Theory of computation~Distributed algorithms}

\keywords{Auditability, Wait-free implementation, Synchronization power, Distributed objects, Shared
memory}

\maketitle

\section{Introduction}

\emph{Auditing} is a powerful tool for determining \emph{who} had access to \emph{which} (potentially sensitive) information. Auditability is crucial for preserving data privacy, as it ensures accountability for data access. This is particularly important in shared, remotely accessed storage systems, where understanding the extent of a data breach can help mitigate its impact.

\subsection{Auditable Read/Write Registers}

Auditability was introduced by Cogo and Bessani~\cite{BessaniDisc} in the context of replicated \emph{read/write registers}. An auditable register extends traditional read and write operations with an additional \emph{audit} operation that reports which register values have been read and by whom. 
The auditability definition by Cogo and Bessani 
is tightly coupled with their multi-writer, multi-reader register 
emulation in a replicated storage system using an information-dispersal scheme.

An implementation-agnostic auditability definition was later proposed~\cite{AttiyaPMPR23}, 
based on collectively linearizing read, write, and audit operations. 
This work also analyzes the consensus number required for implementing auditable 
single-writer registers, showing that it scales with the number of readers and auditors. 
However, this definition assumes that a reader only gains access to values that 
are explicitly \emph{returned} by its read operations. 
This assumption does not account for situations where a reader learns 
the register's value before it has officially returned, 
making the read operation \emph{effective}.
Hence, a notable limitation of this definition is that a process with an effective read can refuse to complete the operation, thereby avoiding detection by the audit mechanism. 

Prior work has also overlooked the risk of non-auditors learning values without explicitly reading them or inferring accesses of other processes. Even when processes follow their prescribed algorithms without active misbehavior, existing auditable register implementations allow an ``honest but curious'' process to learn more than what its read operations officially return. Additionally, extending auditability beyond read/write registers remained an unexplored territory.

\subsection{Our Contributions and Techniques}

In this work, we propose a stronger form of auditability for read/write registers, 
ensuring that all effective reads are auditable and that non-auditors cannot 
infer the values read by other processes. 
We further extend these properties to other data structures 
and propose new algorithms that fulfill these guarantees.

We define new properties that ensure operations do not leak information when processes \ha{are honest-but-curious \cite{DBLP:journals/dc/Goldreich03}}
(see Section~\ref{sec:definitions}).
Firstly, we introduce an implementation-agnostic definition of an \emph{effective operation}, 
which is applicable, for instance, to read operations in an auditable register. 
An operation is effective if a process has determined its return value 
in all executions indistinguishable to it.
Secondly, we define \emph{uncompromised operations}, saying, for example, 
that in a register, readers do not learn which values were read by other readers 
or gain information about values they do not read.
This definition is extended beyond registers.
For arbitrary data objects, we specify that an operation is \emph{uncompromised} if 
there is an indistinguishable execution where the operation does not occur. 

Enforcing uncompromised operations in auditable objects poses a challenge
since it is, in a sense, antithetical to securely logging data accesses.
Our primary algorithmic contribution (Section~\ref{sec:register}) 
is a wait-free, linearizable implementation of an auditable multi-writer, 
multi-reader register. 
Our implementation ensures that all effective reads are auditable 
while preventing information leaks:
reads are uncompromised by other readers, 
and cannot learn previous values --- unless they actually read them. 
As a consequence, the implementation is immune to a \af{honest-but-curious} attacker.

To achieve these properties, our algorithm carefully combines value 
access with access logging. 
Additionally, access logs are encrypted using one-time pads 
known only to writers and auditors.
The subtle synchronization required in our implementation 
is achieved by using \cas{} and \fx{}  
(in addition to ordinary reads and writes).
Such strong synchronization primitives are necessary since even simple single-writer 
auditable registers can solve consensus~\cite{AttiyaPMPR23}.
The correctness proof of the algorithm, 
of basic linearizability properties 
as well as of advanced auditability properties, 
is intricate and relies on a careful linearization function.

Our second algorithmic contribution is an elegant extension of the 
register implementation to other commonly-used objects. 
We first extend our framework to a wait-free, linearizable implementation 
of an auditable multi-writer, multi-reader \emph{max register}~\cite{AspnesAC2012}, 
which returns the largest value ever written.
The semantics of a max register, together with tracking the number of 
operations applied to it (needed for logging accesses), may leak information
to the reader about values it has not effectively read.
We avoid this leakage by adding a \emph{random nonce}, 
serving to introduce some noisiness, to the values written.
(See Section~\ref{sec:maxreg}.)
As before, all effective reads are auditable, 
and no additional information is leaked.

In Section~\ref{sec:snapshot}, we demonstrate how an auditable max register
enables auditability in other data structures. 
Specifically, we implement auditable extension of \emph{atomic snapshots}~\cite{AfekADGMS93} 
and more generally, of \emph{versioned types}~\cite{DenysyukW15}. 
Many useful objects, such as counters and logical clocks, 
are naturally versioned or can be made so with minimal modification.

\subsection{Related Work}

Cogo and Bessani~\cite{BessaniDisc} present an algorithm to 
implement an auditable \emph{regular} register,
using $n \geq 4f+1$ atomic read/write shared objects,
$f$ of which may fail by crashing.
Their high-level register implementation
relies on information dispersal schemes, where the input of a high-level write
is split into several pieces, each written in a different low-level shared object.
Each low-level shared object keeps a trace of each access, and in order to read, 
a process has to collect sufficiently many pieces of information in many low-level 
shared objects, which allows to audit the read.

In asynchronous message-passing systems where $f$ processes can be Byzantine,
Del Pozzo, Milani and Rapetti~\cite{SRDS22} study the possibility of implementing
an atomic auditable register, 
as defined by Cogo and Bessani, with fewer than $4f+1$ servers.
They prove that without communication between servers,
auditability requires at least $4f+1$ servers, $f$ of which may be Byzantine.
They also show that allowing servers to communicate with each other admits
an auditable atomic register with optimal resilience of $3f+1$. 

Attiya, Del Pozzo, Milani, Pavloff and Rapetti~\cite{AttiyaPMPR23} 
provides the first implementation-agnostic auditability definition. 
Using this definition they show that auditing adds power to reading and writing, 
as it allows processes to solve consensus, 
implying that auditing requires strong synchronization primitives.
They also give several implementations that use non-universal primitives
(like swap and fetch\&add), for a single writer 
and either several readers or several auditors 
(but not both). 

When faulty processes are \emph{malicious},
\emph{accountability}~\cite{peerreview,Civit22Crime, DBLP:journals/jpdc/CivitGGGK23, DR2022Tenderbake}
aims to produce proofs of misbehavior in instances where processes deviate, 
in an observable way, from the prescribed protocol.
This allows the identification and removal of malicious processes from the system
as a way to clean the system after a safety violation.
In contrast, auditability logs the processes' actions and lets the auditor
derive conclusions about the processes' behavior.

In addition to tracking access to shared data, it might be desirable to
give to some designated processes the ability to grant and/or revoke
access rights to the data. Frey, Gestin and Raynal~\cite{DBLP:conf/wdag/FreyGR23}
specify and investigate the synchronization power of shared objects
called  \emph{AllowList} and \emph{DenyList}, allowing a set of
manager processes  to grant or revoke access rights for a given set of
resources.

 \section{Definitions}
\label{sec:definitions}

\paragraph{Basic notions}

We use a standard model, in which a set of processes $p_1, \dots, p_n$, 
communicate through a shared memory consisting of \emph{base objects}.
The base objects are accessed with \emph{primitive operations}.
In addition to atomic reads and writes, our implementations use 
two additional standard synchronization primitives:
$\cas(R,old,new)$ atomically compares the current value of $R$ with 
$old$ and if they are equal, replaces the current value of $R$ with $new$;
$\fx(R,arg)$ atomically replaces the current value of $R$ with a 
bitwise XOR of the current value and $arg$.\footnote{\fx{} 
is part of the ISO C++ standard since C++11~\cite{CPP}.} 
 
An \emph{implementation} of a (high-level) object $T$ specifies a program for each process and each operation of the object $T$;
when receiving an \emph{invocation} of an operation,
the process takes \emph{steps} according to this program. 
Each step by a process consists of some local computation,
followed by a single primitive operation on the shared memory.
The process may change its local state after a step, and it
may return a \emph{response} to the operation of the high-level object.

Implemented (high-level) operations are denoted with capital letters, 
e.g., \rd{}, \wrt{}, \adt{}, 
while primitives applied to base objects, 
appear in normal font, e.g., read and write.

A \emph{configuration} $C$ specifies the state of every process and of every base object.
An \emph{execution} $\alpha$ is an alternating sequence of configurations and events, starting with an \emph{initial configuration}; it can be finite or infinite.
For an execution $\alpha$ and a process $p$, 
$\alpha |_p$ is the projection of $\alpha$ on events by $p$.
For two executions $\alpha$ and $\beta$, we write 
$\ind{\alpha}{p}{\beta}$ when $\alpha |_p = \beta|_p$, 
and say that $\alpha$ and $\beta$ are \textit{indistinguishable} to process $p$.

An operation $op$ \emph{completes} in an execution $\alpha$ if $\alpha$ 
includes both the invocation and response of $op$;
if $\alpha$ includes the invocation of $op$, 
but no matching response, then $op$ is \emph{pending}.
An operation $op$ \emph{precedes} another operation $op'$ in $\alpha$
if the response of $op$ appears before the invocation of $op'$ in $\alpha$.

A \emph{history} $H$ is a sequence of invocation and response events;
no two events occur at the same time.
The notions of \emph{complete}, \emph{pending} and \emph{preceding} operations
extend naturally to histories. 

The standard correctness condition for concurrent implementations 
is \emph{linearizability}~\cite{HerlihyWing90}: intuitively,
it requires that each operation appears to take place instantaneously 
at some point between its invocation and its response.
Formally:

\begin{definition}
Let \m{A} be an implementation of an object $T$.
An execution $\alpha$ of \m{A} is \emph{linearizable} if there is 
a sequential execution $L$ (a \emph{linearization} of the operations on $T$ in $\alpha$) 
such that:
\begin{itemize}
    \item $L$ contains all complete operations in $\alpha$, 
    and a (possibly empty) subset of the pending operations in $\alpha$ (completed with response events),
    \item 
    If an operation $op$ precedes an operation $op'$ in $\alpha$, 
    then $op$ appears before $op'$ in $L$, and
    \item $L$ respects the sequential specification of the high-level object.
\end{itemize}
\m{A} is \emph{linearizable} if all its executions are linearizable. 
\end{definition}


An implementation is \emph{lock-free} if, whenever there is a pending operation, 
some operation returns in a finite number of steps of all processes.
Finally, an implementation is \emph{wait-free} if, whenever there is 
a pending operation by process $p$, 
this operation returns in a finite number of steps by $p$.

\paragraph{Auditable objects}

An auditable register supports, in addition to the standard 
$\rd{}$ and $\wrt{}$ operations, 
also an $\adt{}$ operation that reports which values were read by each process. 
Formally, an \adt{} has no parameters and it returns a set of pairs, 
$(j,v)$, where $j$ is a process id, and $v$ is a value of the register.
A pair $(j,v)$ indicates that process $p_j$ has read the value $v$.

Formally, the sequential specification of an auditable register enforces, 
in addition to the requirement on \rd{} and \wrt{} operations, 
that a pair appears in the set returned by an \adt{} operation if and only if 
it corresponds to a preceding \rd{} operation.
In prior work~\cite{AttiyaPMPR23}, 
this \emph{if and only if} property was stated as a 
combination of two properties of the sequential execution: 
\emph{accuracy}, if a \rd{} is in the response set of the \adt{},
then the \rd{} is before the \adt{} (the \emph{only if} part), 
and \emph{completeness}, any \rd{} before the \adt{} is in its response set 
(the \emph{if} part).

We wish to capture in a precise, 
implementation-agnostic manner, 
the notion of an \emph{effective operation},
which we will use to ensure that an \adt{} operation will report 
all \emph{effective} operations. 
Assume an algorithm $\m{A}$ that implements an object $T$.
The next definition characterizes, 
in an execution in which a process $p$ invokes an operation, 
a point at which $p$ knows the value that the operation returns, 
even if the response event is not present.

\begin{definition}[effective operation]
\label{def:non_effective}
An operation $op$ on object $T$ by process $p$ is 
$v$-effective after a finite execution prefix $\alpha$ if, 
for every execution prefix $\beta$ indistinguishable from $\alpha$ to $p$ 
(i.e., such that $\ind{\alpha}{p}{\beta}$), $op$ returns $v$ in every extension $\beta'$ of $\beta$ in which $op$ completes.
\end{definition}

Observe that in this definition, 
$\alpha$ itself is also trivially an execution prefix indistinguishable to $p$, 
and hence in any extension $\alpha'$ in which $op$ completes returns value $v$. 
Observe as well that $op$ could already be completed in $\alpha$ or not be invoked (yet). 
However, the most interesting case is when $op$ is pending in $\alpha$.


We next define the property that an operation on $T$ 
is not compromised in an execution prefix by a process. 
As we will see, in our register algorithm, 
a \rd{} by $p$  is linearized as soon as it becomes $v$-effective, 
in a such way that in any extension including a complete \adt{}, 
$p$ is reported as a reader of $v$ by this \adt{}. 
This, however, does not prevent a curious reader $p$ from learning another 
value $v'$ for which none of its \rd{} operations is $v'$-effective. 
In such a situation, the \wrt{} operation with input $v'$ 
is said to be \emph{compromised} by $p$.
The next definition states that this can happen only if 
a \rd{} operation by $p$ becomes $v'$-effective. 
The definition is general, and applies to any object.

\begin{definition}[uncompromised operation]
\label{def:uncompromised_op}
Consider a finite execution prefix $\alpha$ and an operation $op$ 
by process $q$ whose invocation is in $\alpha$.
We say that $op$ is \emph{uncompromised in $\alpha$ by process $p$} if there is 
another finite execution $\beta$ such that $\ind{\alpha}{p}{\beta}$ 
and $op$ is not invoked in $\beta$.
\end{definition}

\ha{A value $v$ is \emph{uncompromised by a reader $p$} if all \wrt{}($v$) operations are uncompromised by $p$, 
unless $p$ has an effective \rd{} returning $v$.}



\paragraph{One-time pads}

To avoid data leakage, we employ \emph{one-time pads} \cite{miller2024telegraphic,vernam}.
Essentially, a one-time pad is a random string---known only to 
the writers and auditors---with a bit for each reader. 
To encrypt a message $m$, $m$ is bitwise XORed with the pad obtaining a ciphertext $c$.
Our algorithm relies on an infinite sequence of one-time pads.
A one-time pad is \emph{additively malleable}, i.e., 
when $f$ is an additive function, it is possible 
to obtain a valid encryption of $f(m)$ by applying a corresponding 
function $f'$ to the ciphertext $c$ corresponding to $m$.

\paragraph{Attacks}
We consider an \af{honest-but-curious (aka, semi-honest and passive) \cite{DBLP:journals/dc/Goldreich03}} \ha{
attacker 
that} interacts with the 
implementation of $T$ by performing operations, and adheres to its code. 
It may however stop prematurely and perform arbitrary local computations 
on the responses obtained from base objects. 
For instance, for an auditable register, 
the attacker can attempt to infer in a \rd{} operation the current 
or a past value of the register, without being reported in \adt{} operations.

\section{An Auditable Multi-writer, Multi-reader Register}
\label{sec:register}

We present a wait-free and linearizable implementation of 
a multi-writer, multi-reader register (Alg.~\ref{alg:cas_mwa}), 
in which effective reads are auditable. 
Furthermore, the implementation does not compromise other reads, 
as while performing a read operation, a process is neither able to learn previous values, 
nor whether some other process has read the current value. 
We ensure that a read operation is linearized as soon as, 
and not before it becomes effective. 
Audits hence report exactly those reads that have made enough progress 
to infer the current value of the register. 
As a consequence, the implementation is immune to \af{an
honest-but-curious} attacker.


\subsection{Description of the Algorithm}
\label{subsec:register_descrption}

The basic idea of the implementation is to store in a single register $R$,
the current value and a sequence number, 
as well as the set of its readers, encoded as a bitset. 
Past values, as well as their reader set, 
are stored in other registers (arrays $V$ and $B$ in the code, 
indexed by sequence numbers), so auditors can retrieve them. 
Changing the current value from $v$ to $w$ consists in first copying $v$ and 
its reader set to the appropriate registers $V[s]$ and $B[s]$,
respectively  (where $s$ is $v$'s sequence number), 
before updating $R$ to a triple formed by $w$, a new sequence number,
and an empty reader set. 
This is done with a \cas{} in order not to miss changes to the reader set 
occurring between the copy and the update.  
An auditor starts by reading $R$, obtaining the current value $w$, 
its set of readers, and its sequence number $s$. 
Then it goes over arrays $B$ and $V$ to retrieve previous values 
written and the processes that have read them. 

In an initial design of the implementation, a \rd{} operation obtains from $R$ the current value $v$ and the reader set, 
adding locally the ID of the  reader to this set before writing it back to $R$, 
using \cas{}. 
This simple design is easy to linearize (each operation is linearized with 
a \cas{} or a $\mathsf{read}$ applied to $R$). 
However, besides the fact that \rd{} and \wrt{} are only lock-free, 
this design has two drawbacks regarding information leaking:

\emph{First}, a reader can read the current value without being reported by \adt{} operations,
simply by not writing to the memory after reading $R$, 
when it already knows the current value $v$ of the register. 
This step does not modify the state of $R$ (nor of any other shared variables), 
and it thus cannot be detected by any other operation. 
Therefore, by following its code, but pretending to stop immediately after accessing $R$, 
a reader is able to know the current value without ever being reported by \adt{} operations. 

\emph{Second}, each time  $R$ is read by some process $p$, 
it learns which readers have already read the current value. 
Namely, while performing a \rd{} operation, 
a process can compromise other reads. 

Alg.~\ref{alg:cas_mwa} presents the proposed implementation of an auditable register.
We deflect the ``crash-simulating'' attack by having each \rd{} operation apply at most one primitive to $R$ that atomically returns the content of $R$ and updates the reader set.  To avoid partial auditing, the reader set is encrypted, while still permitting insertion by modifying the encrypted set (i.e., a light form of homomorphic encryption.).  Inserting the reader ID into the encrypted set should be kept simple, as it is part of an atomic modification of $R$. We apply to the reader set a simple cipher (the one-time pad~\cite{miller2024telegraphic,vernam}), 
and benefit from its additive malleability. 
Specifically, the IDs of the readers of the current value are tracked by the last $m$ bits of $R$, where $m$ is the number of readers. When a new value with sequence number $s$ is written in $R$, these bits are set to a random $m$-bit string, $\mathit{rand}_s$, only known by writers and auditors. This corresponds to encrypting the empty set with a random mask. 
Process $p_i$ is inserted in  the set by XORing the $i$th tracking bit with $1$. Therefore, retrieving the value stored in $R$ and updating the reader set can be done atomically by applying \fx{}. Determining set-membership requires the mask $\mathit{rand}_s$, known only to auditors and writers. 

The one-time pad, as its name indicates, is secure as long as each mask is used at most once. 
This means we need to make sure that different sets encrypted with the same mask $\mathit{rand}_s$ 
are never observed by a particular reader, otherwise,
the reader may infer some set member by XORing the two ciphered sets.
To ensure that, we introduce an additional register $\emph{SN}$, 
which stores only the sequence number of the current value. 
A \rd{} operation by process $p_i$ starts by reading $\emph{SN}$, 
and, if it has not changed since the previous \rd{} by the same process, 
immediately returns the latest value read. 
Otherwise, $p_i$ obtains the current value $v$ and records itself as one of its readers 
by applying a \fx($2^i$) operation to $R$. 
This changes the $i$th tracking bit, leaving the rest of $R$ intact. 
Finally, $p_i$ updates $\emph{SN}$ to the current sequence number read from $R$, 
thus ensuring that $p_i$ will not read $R$ again, unless its sequence number field is changed.  
This is done with a \cas{} to avoid writing an old sequence number in $SN$.

Writing a new value $w$ requires retrieving and storing the IDs of the readers of the current 
value $v$ for future \adt{}, writing $w$, the new sequence number $s+1$, and an empty reader set encrypted with a fresh mask $\mathit{rand}_{s+1}$ to $R$ before announcing the new sequence number in $\mathit{SN}$. To that end, $p_j$ first locally gets a new sequence number
$s+1$, where $s$ is read from $SN$. It then repeatedly reads $R$, deciphers the tracking bits and updates shared registers $V[s]$ and $B[s]$ accordingly until it succeeds in changing it to $(s+1,w,\mathit{rand}_{s+1})$ or it discovers a sequence number $s' \geq s+1$ in $R$. In the latter case, a concurrent \wrt{}($w'$) has succeeded, and may be seen as occurring immediately after $p_j$'s operation, which therefore can be abandoned. In the absence of a concurrent \wrt{}, the \cas{} applied to $R$ may fail as  the tracking bits are modified by a concurrent \rd{}. This happens at most $m$ times, as each reader applies at most one \fx{} to $R$ while its sequence number field does not change. Whether or not $p_j$ succeeds in modifying $R$, we make sure that before \wrt{}($w$) terminates, the sequence number $SN$ is at least as large as the new sequence number  $s+1$. In this way, after that, \wrt{} operations overwrite the new value $w$ and \rd{} operations return $w$ or a more recent value.

Because $\mathit{SN}$ and $R$ are not  updated atomically, their sequence number fields may differ. 
In fact, an execution of Alg.~\ref{alg:cas_mwa} alternates between 
\emph{normal} $E$ phases, in which both sequence numbers are \emph{equal}, 
and \emph{transition} $D$ phases in which they \emph{differ}. 
A transition phase is triggered by a \wrt{}($w$) with sequence number $s$ and 
ends when the \wrt{} completes or it is helped to complete by updating  $SN$ to $s$. 
Care must be taken during a $D$ phase, as some \rd{}, which is \emph{silent}, 
may return the old value $v$, 
while another, \emph{direct}, \rd{} returns the value $w$ being written. 
For linearization, we push back silent \rd{} before the \cas{} applied to $R$ 
that marks the beginning of phase $D$, 
while a direct \rd{} is linearized with its \fx{} applied to $R$.

An \adt{}  starts by reading $R$, thus obtaining the current value $v$, and its sequence number $s$;
it is linearized with this step. 
It then returns the set of readers for  $v$ (inferred from the tracking bits read from $R$) 
as well as  for each previously written value  (which can be found in the registers $V[s']$ and $B[s']$, for $s'<s$.). 
In a $D$ phase, a silent \rd{} operation may start after an \adt{} reads $R$ while being linearized before this step, 
so we make sure that the $D$ phase ends before  the \adt{} returns. 
This is done, as in direct \rd{} and \wrt{}, by making sure that $\mathit{SN}$ is at least as large 
as the sequence number $s$ read from $R$. 
In this way, a silent \rd{} (this  also holds for a \wrt{} that is immediately overwritten)  
whose linearization point is pushed back before that of an \adt{} is concurrent with this \adt{}, 
ensuring that the linearization order respects the real time order between these operations.

Suppose that an \adt{} by some process $p_i$ reports $p_j$ as a reader of some value $v$.  
This happens because $p_i$  directly identifies $p_j$ as a reader of $v$ from the tracking bits in $R$, 
or indirectly by reading the registers $V[s]$ and $B[s]$, where $V[s] = v$. 
In both cases, in a \rd{} instance $op$, reader $p_j$ has previously applied a \fx{} to $R$ while its value field is $v$. 
Since the response of this \fx{} operation completely determines the return value of $op$, 
independently of future or past steps taken by $p_j$, 
$op$ is effective. 
Therefore, only effective operations are reported by \adt{}, 
and if an \adt{} that starts after $op$ is effective, 
it will discover that $p_j$ read $v$, 
again either directly in the tracking bits of $R$, or indirectly 
after the reader set has been copied to $B[s]$.

\begin{algorithm}[tbp]
  \begin{algorithmic}[1]
    \Statex \textbf{shared registers:}
    \Statex\hspace{\algorithmicindent} $\mathit{R}$: a register supporting $\mathsf{read}$, \cas{}, and \Statex\hspace{\algorithmicindent}\hspace{\algorithmicindent} \fx{}, initially $(0,v_0,\mathit{rand}_0)$
    \Statex \Comment{store a triple (sequence number, value, $m$-bits string)}
    \Statex\hspace{\algorithmicindent} $\mathit{SN}$: a register supporting $\mathsf{read}$ and \cas, 
    \Statex\hspace{\algorithmicindent}\hspace{\algorithmicindent} initially $0$ 
    \Statex\hspace{\algorithmicindent} $V[0..+\infty]$ 
    registers, initially $[\bot,\ldots,\bot]$
    \Statex\hspace{\algorithmicindent} $B[0..+\infty][0..m-1]$ 
    Boolean registers, initially, 
    \Statex\hspace{\algorithmicindent}\hspace{\algorithmicindent} $B[s,j] = \emph{false}$ for every $(s,j) : s  \geq 0, 0 \leq j < m$.  
    
    \Statex \textbf{local variables: reader}
    \Statex\hspace{\algorithmicindent} $prev\_val, prev\_sn$: latest value read ($\bot$ initially) 
    \Statex\hspace{\algorithmicindent}\hspace{\algorithmicindent} and its sequence number ($-1$ initially)
    \Statex \textbf{local variables common to writers and auditors}
    \Statex\hspace{\algorithmicindent} $\mathit{rand}_0, \mathit{rand}_1, \ldots$: 
    sequence of random $m$-bit strings 
    \Statex \textbf{local variables: auditor}
    \Statex\hspace{\algorithmicindent}  $A$: audit set, initially $\emptyset$; \Statex\hspace{\algorithmicindent}  $lsa$: latest ``audited'' seq. number, initially $0$

    \Function{read}{$\,$} \Comment{code for reader $p_j$, $0 \leq j < m$}
    \State $sn \gets \mathit{SN}.\mathsf{read}()$ \label{line:mwmrma:read_sn}
    \State \textbf{if} $sn = prev\_sn$ \textbf{then} \Return{$prev\_val$} \label{line:mwmrma:no_new_write}\label{line:mwmrma:return_prev_val} 
    \Statex\Comment{no new write since latest \rd{} operation}
    \State  $(sn, val, \_) \gets R.\fx(2^j)$ \label{line:mwmrma:read_fetch} 
    \Statex\Comment{fetch current value and insert $j$ in reader set}
    \State $SN.\cas(sn-1, sn)$\label{line:mwmrma:reader_cas}  
  \Comment{help complete $sn$th \wrt{}}
    \State $prev\_sn \gets sn$; $prev\_val \gets val$; \Return{$val$} \label{line:mwmrma:return_val}
    \EndFunction

    \Function{write}{$v$} \Comment{code for writer $p_i, i \notin \{0,\ldots,m-1\}$}
    \State $sn \gets \mathit{SN}.\mathsf{read}()+1$ \label{line:mwmrma:writer_read_sn}
    \Repeat\label{line:mwmrma:writer_repeat}
    \State $(lsn,lval, bits) \gets R.\mathsf{read}()$ \label{line:mwmrma:writer_read_R}
    \State \textbf{if} $lsn \geq sn$ \textbf{then} \textbf{break} \label{line:mwmrma:writer_break}
    \State $V[lsn].\mathsf{write}(lval)$; 
    \State \textbf{for each} $j : bits[j] \neq \mathit{rand}_{lsn}[j]$ \textbf{do} \Statex\hspace{\algorithmicindent} $B[lsn][j].\mathsf{write}(true)$ \label{line:mwmrma:writer_record_reads}
    \Until{$R.\cas((lsn,\mathit{lval},bits),(sn,v,\mathit{rand}_{sn}))$}
    \label{line:mwmrma:writer_cas}
    \State $SN.\cas(sn-1,sn)$; \Return{} \label{line:mwmrma:writer_return}
\EndFunction

\Function{audit}{$\,$}
\State $(rsn,rval,rbits) \gets R.\mathsf{read}()$  \label{line:mwmrma:audit_read}
\For{$s = lsa,lsa+1,\ldots,rsn-1$}\label{line:mwmrma:audit_for}
\State $val \gets V[s].\mathsf{read}()$;
\State $A \gets A \cup \{(j,val):  0 \leq j < m ,  B[s][j].\mathsf{read}() = \emph{true}\}$\label{line:mwmrma:audit_definitive} 
\EndFor
\State $A \gets A \cup \{(j,rval) : 0 \leq j < m, bits[j] \neq \mathit{rand}_{rsn}[j]\}$\label{line:mwmrma:audit_non_definitive}
\State $lsa \gets sn$; $SN.\mathsf{\cas}(rsn-1,rsn)$;  \Return{$A$}\label{line:mwmrma:audit_return}
\EndFunction

\end{algorithmic}
\caption{Multi-writer, $m$-reader auditable register implementation
}
\label{alg:cas_mwa}
\end{algorithm}

\subsection{Proof of Correctness}
\label{subsec:register_proof}

\paragraph{Partitioning into phases}

We denote by $R.seq, R.val$ and $R.bits$ the sequence number, 
value and $m$-bits string, respectively, stored in $R$.
We start by observing that the pair of values in $(R.seq,SN)$ 
takes on the following sequence: $(0,0), (1,0), (1,1), \ldots, (x,x-1), (x,x), \ldots$ 
Indeed, when the state of the implemented register changes to a new value $v$, 
this value is written to $R$ 
together with a sequence number $x+1$, 
where $x$ is the current value of $SN$. 
$SN$ is then updated to $x+1$, and so on.


Initially, $(R.seq,SN) = (0,0)$. 
By invariants that can be proved on the algorithm, 
the successive values of $R.seq$ and $SN$ are $0,1,2,\ldots$,
$SN \geq x-1$ when $R.seq$ is changed to $x$, 
and when $SN$ is changed to $x$, $R.seq$ has previously been updated to $x$. 
Therefore, the sequence of successive values of the pair $(R.seq,SN)$ 
is $(0,0), (1,0), (1,1), \ldots,$ $(x,x-1),(x,x), \ldots$. 
We can therefore partition any execution into intervals 
$E_x$ and $D_x$ (for $E$qual and $D$ifferent), so that $R.seq =x$ and $SN=x$ during $E_x$,
and $R.seq =x$ and $SN=x-1$ during $D_x$:


\begin{lemma}
\label{lem:epoch}
A finite execution $\alpha$ can be written, for an integer $k \geq 0$, 
either as $E_0 \rho_1 D_1 \sigma_1 E_1 \ldots \rho_k D_k \sigma_k E_k$ 
or as $E_0 \rho_1 D_1 \sigma_1 E_1 \ldots$ $\sigma_{k-1}$ $E_{k-1} \rho_k D_k$,
where:
  \begin{itemize}
  \item $\rho_\ell$ and $\sigma_\ell$ are the steps that respectively 
  change the value of $R.seq$ and $SN$  from $\ell-1$ to $\ell$ 
  ($\rho_\ell$ is a successful $R.\cas{}$, line~\ref{line:mwmrma:writer_cas}, 
  $\sigma_\ell$ is also a successful $\mathit{SN}.\cas{}$, applied within a \rd{}, line~\ref{line:mwmrma:reader_cas}, 
  a \wrt{}, line~\ref{line:mwmrma:writer_return}, or an \adt{}, line~\ref{line:mwmrma:audit_return}).
  \item in any configuration in $E_\ell$, $R.seq = SN = \ell$,  and in any configuration in $D_\ell$, $R.seq = \ell = SN +1$. 
  \end{itemize}
\end{lemma}

\paragraph{Termination}

It is clear that \adt{} and \rd{} operations are wait-free. 
We prove that \wrt{} operations are also wait-free, by showing that 
the repeat loop (lines~\ref{line:mwmrma:writer_repeat}-\ref{line:mwmrma:writer_cas}) 
terminates after at most $m+1$ iterations. 
This holds since each reader may change $R$ at most once 
(by applying a $R.\fx{}$, line~\ref{line:mwmrma:read_fetch}) 
while $R.seq$ remains the same.

\begin{restatable}{lemma}{WaitFree}
  \label{lem:cas_mwa_wf}
  Every operation terminates within a finite number of its own steps. 
\end{restatable}

\begin{proof}[Proof sketch]
  The lemma clearly holds for \rd{} and \adt{} operations.
  Let $wop$ be a \wrt{} operation, and assume, towards a contradiction, that it does not terminate. Let $sn = x+1$ be the sequence number obtained at the beginning of $wop$ at line~\ref{line:mwmrma:writer_read_sn}, where $x$ is the value read from $SN$.  We denote by $(sr,vr,br)$  the triple read from $R$ in the first iteration of the repeat loop. 
  It can be shown that $x \leq sr$.  As $sr < sn = x+1$ (otherwise the loop breaks in the first iteration at line~\ref{line:mwmrma:writer_break}, and the operation terminates),  we have $sr = x$.

  As $wop$ does not terminate, in particular the \cas{} applied to $R$ at the end of the first iteration fails. Let $(sr',vr',br')$ be the value of $R$ immediately before this step is applied. 
  This can be used to show that if $sr' \neq sr$ or $vr' \neq vr$, 
  then $sr' > sr$. Therefore, $wop$ terminates in the next iteration as the sequence number read from $R$ in that iteration is greater than or equal to $sn$ (line~\ref{line:mwmrma:writer_break}). It thus follows that $sr = sr'$, $vr = vr'$, and  $br \neq br'$: at least one reader applies a \fx{} to $R$ during the first iteration of repeat loop. 

  The same reasoning applies to the next iterations of the repeat loop. In each of them, the sequence number and the value stored in  $R$ are the same, $sr$ and $vr$ respectively (otherwise the loop would break at line~\ref{line:mwmrma:writer_break}), and thus a reader applies a \fx{} to $R$ before the \cas{} of line~\ref{line:mwmrma:writer_cas} (otherwise the \cas{} succeeds and $wop$ terminates). But 
  it can be shown that each reader applies at most one \fx{} to $R$ while it holds the same sequence number, which is a contradiction.
\end{proof}

\paragraph{Linearizability}

Let $\alpha$ be a finite execution, and $H$ be the history of the \rd{}, \wrt{}, and \adt{} operations in $\alpha$.
We classify and associate a sequence number with some of \rd{} and \wrt{} operations in $H$ 
as explained next.
Some operations that did not terminate are not classified,
and they will later be discarded.
\begin{itemize}
\item A \rd{} operation $op$ is \emph{silent} if it reads $x = prev\_sn$ at line~\ref{line:mwmrma:read_sn}. 
  The sequence number $sn(op)$ associated with a \emph{silent} \rd{} operation $op$ is the value $x$ returned by the read from $\mathit{SN}$.
  Otherwise, if $op$ applies a \fx{} to $R$, it is said to be \emph{direct}. Its sequence number $sn(op)$ is the one fetched from $R$ (line~\ref{line:mwmrma:read_fetch}). 
\item A \wrt{} operation $op$ is \emph{visible} if it applies a successful \cas{} to $R$ (line~\ref{line:mwmrma:writer_cas}). Otherwise, if $op$ terminates without applying a successful \cas{} on $R$ (by exiting the repeat loop from the break statement, line~\ref{line:mwmrma:writer_break}), it is said to be \emph{silent}. For both cases, the sequence number $sn(op)$ associated with $op$  is $x+1$, where $x$ is the value read from $SN$ at the beginning of $op$ (line~\ref{line:mwmrma:writer_read_sn}). 
\end{itemize}
Note that all terminated \rd{} or \wrt{} operations are classified as silent, direct, or visible. 
An \adt{} operation $op$ is associated with the sequence number read from $R$ at line~\ref{line:mwmrma:audit_read}. 

We define a complete  history $H'$ by removing or completing the operations that do not terminate in $\alpha$, as follows:
Among the operations that do not terminate, we remove every \adt{} and every unclassified \rd{} or \wrt{}. 
For a silent \rd{} that does not terminate in $\alpha$, 
we add a response immediately after $\emph{SN}$ is read at line \ref{line:mwmrma:read_sn}. 
The value returned is $prev\_val$, that is the value returned by the previous \rd{} by the same process. 
For each direct \rd{} operation $op$ that does not terminate in $\alpha$, 
we add a response with value $v$ defined as follows. 
Since $op$ is direct, it applies a \fx{} on $R$ that returns a triple $(sr,vr,br)$; $v$ is the value $vr$ in that triple. In $H'$, we place the response of non-terminating direct \rd{} and visible \wrt{} after every response and every remaining invocation of $H$, in an arbitrary order. 

Finally, to simplify the proof, we add at the beginning of $H'$ an invocation immediately followed by a response of a \wrt{} operation with input $v_0$ 
(the initial value of the auditable register.). 
This fictitious operation has sequence number $0$ and is visible.

Essentially, in the implemented register updating to a new value $v$  is done in two phases. $R$ is first modified to store $v$ and a fresh sequence number $x+1$, and then the new sequence number is announced in $SN$. Visible \wrt{}, direct \rd{}, and \adt{} operations may be linearized with respect to the
\cas{}, \fx{} or $\mathsf{read}$ they apply to $R$. 
Special care should be taken for silent \rd{} and \wrt{} operations. Indeed, a silent \rd{} that reads $x$ from $SN$, may return  the previous value $u$  stored in the implemented register or $v$, depending on the sequence number of the last preceding direct \rd{} by the same process. Similarly, a silent \wrt{}($v'$) may not access $R$ at all, or apply a \cas{} after $R.seq$ has already been changed to $x+1$. However,  \wrt{}($v'$) has to be linearized before  \wrt{}($v$), in such a way that $v'$ is immediately overwritten. 

Hence, direct \rd{}, visible \wrt{}, and \adt{} 
are linearized first, according to the order in which they apply a primitive to $R$. We then place the remaining operations with respect to this partial linearization.  $L(\alpha)$ is  the total order on the operations in $H'$ obtained by the following rules:
\begin{enumerate}
\item[R1] For direct \rd{}, visible \wrt{},  \adt{} and some silent \rd{} operations we defined an associated step $ls$ applied by the operation. These operations are then ordered according to the order in which their associated step takes place in $\alpha$. For a direct \rd{}, visible \wrt{}, or  \adt{} operation $op$, its associated step $ls(op)$ is respectively  the \fx{} at line~\ref{line:mwmrma:read_fetch}, the successful \cas{} at  line~\ref{line:mwmrma:writer_cas}, and the $\mathsf{read}$ at line~\ref{line:mwmrma:audit_read} applied to $R$. For a silent \rd{} operation $op$ with sequence number $sn(op) = x$, if $SN.\mathsf{read}$ (line~\ref{line:mwmrma:read_sn}) is applied in $op$ during  $E_x$ (that is, $R.seq = x$ when this read occurs), $ls(op)$ is this $\mathsf{read}$ step. The other silent \rd{} operations do not have a linearization step, and are not ordered by this rule. They are instead linearized by Rule R2.
\end{enumerate}
Recall that $\rho_{x+1}$ is the successful \cas{} applied to $R$ that changes $R.seq$ from $x$ to $x+1$ (Lemma~\ref{lem:epoch}). 
By rule R1, the visible \wrt{} with sequence number $x+1$ is linearized at $\rho_{x+1}$.  
\begin{enumerate}
\item[R2] For every $x \geq 0$,  every remaining silent \rd{} 
  $op$ with sequence number $sn(op) = x$ is placed immediately before the unique visible \wrt{} operation with sequence number $x+1$. Their relative order follows  the order in which their $\mathsf{read}$ step of $SN$ (line~\ref{line:mwmrma:read_sn}) 
  is applied in $\alpha$.
\item[R3] Finally, we place for each $x \geq 0$ every silent \wrt{} operation $op$ with sequence number $sn(op) =  x+1$.  They are placed after the silent \rd{} operations with sequence number $x$ ordered according to rule R2, and before the unique  visible \wrt{} operation with sequence number $x+1$. As above, their respective order is determined by the order in which  their $\mathsf{read}$ step of $SN$ (line~\ref{line:mwmrma:writer_read_sn}) is applied in $\alpha$. 
\end{enumerate}

Rules R2 and R3 are well-defined, 
is we can prove the existence and uniqueness 
of a visible \wrt{} with sequence number $x$, 
if there is an operation $op$ with $sn(op) = x$. 

We can show that the linearization $L(\alpha)$ extends the real-time order 
between operations, 
and that the \rd{} and \wrt{} operations satisfy the sequential 
specification of a register. 

\paragraph{Audit Properties}
For the rest of the proof, fix a finite execution $\alpha$.
The next lemma helps to show that effective operations are audited;
it demonstrates how indistinguishability is used in our proofs.

\begin{lemma}
\label{lem:eff_read_lin}
A \rd{} operation $rop$ that is invoked in $\alpha$ is in $L(\alpha)$ 
if and only if $rop$ is effective in $\alpha$. 
\end{lemma}

\begin{proof}
If $rop$ completes in $\alpha$, then it is effective and it is in $L(\alpha)$. 
Otherwise, $rop$ is pending after $\alpha$.
Let $p_j$ be the process that invokes $rop$.
We can show:

\begin{claim}
$rop$ is effective after $\alpha$ if and only if either\\ 
(1) $p_j$ has read $x$ from $\mathit{SN}$ and $x = prev\_sn$ (line~\ref{line:mwmrma:read_sn}) or \\
(2) $p_j$ has applied  \fx{} to $R$ (line~\ref{line:mwmrma:read_fetch}). 
\end{claim}

\begin{proof}
First, 
let $\alpha'$ be an arbitrary extension of $\alpha$ in which $rop$ returns some value $a$, 
$\beta$ a finite  execution indistinguishable from $\alpha$ to $p_j$, 
and $\beta'$ one of its extensions in which $rop$ returns some value $b$. 
We show that if $\alpha$ satisfies (1) or (2), then $a = b$. 
(1) If in $\alpha$ after invoking $rop$, $p_j$ reads $x = prev\_sn$ from $\mathit{SN}$ at line~\ref{line:mwmrma:read_sn}, 
then $rop$ returns $a = prev\_val$ in $\alpha'$. 
Since $\ind{\alpha}{p_j}{\beta}$, $prev\_val = a$ and $prev\_sn = x$ when $rop$ starts in $\beta$,  
and $p_j$ reads also $x$ from $SN$. 
Therefore, $rop$ returns $b = a$ in $\beta'$. 
(2) If $p_j$ applies a \fx{} to $R$ (line~\ref{line:mwmrma:read_fetch}) while performing $rop$ in $\alpha$,
then $rop$ returns $a = v$ (line~\ref{line:mwmrma:return_val}), where $v$ is the value fetched from $R.val$ in $\alpha'$. 
Since  $\ind{\alpha}{p_j}{\beta}$, $p_j$ also applies a \fx{} to $R$ while performing $rop$ in $\beta$, 
and fetches $v$ from $R.val$. Therefore $rop$ also returns $v$ in $\beta'$.

Conversely, suppose that neither (1) nor (2) hold for $\alpha$. 
That is, $p_j$ has not applied a \fx{} to $R$ and, 
if $x$ has been read from $\mathit{SN}$, $x \neq prev\_sn$. 
We construct two extensions $\alpha'$ and $\alpha''$ in which $rop$ returns $v' \neq v''$, respectively. 
Let $X$ be the value of $SN$ at the end of $\alpha$, and $p_i$ be a writer. 
In $\alpha'$, $p_i$ first completes its pending \wrt{} if it has one, 
before repeatedly writing the same value $v'$ until performing a visible \wrt{}($v'$). 
Finally, $p_j$ completes $rop$. 
Since $p_i$ is the only writer that takes steps in $\alpha$, 
it eventually has a visible \wrt{}($v'$), 
that is in which $R.val$ is changed to $v'$. 
Note also that when this happens, $\emph{SN} > X$. 
The extension $\alpha''$ is similar, except that $v'$ is replaced by $v''$. 

Since conditions (1) and (2) do not hold, 
$p_i$'s next step in $rop$ is reading $SN$ or issuing  $R.\fx{}$.  
If $p_j$ reads $SN$ after resuming $rop$, it gets a value $x > prev\_val$. 
Thus, in both cases, $p_j$ accesses $R$ in which it reads $R.val = v'$ (or $R.val = v''$). 
Therefore, $rop$ returns $v'$ in $\alpha'$ and $v''$ in $\alpha''$.
\end{proof}

Now, if (1) holds ($p_j$  reads $x = prev\_val$ from $SN$ at line~\ref{line:mwmrma:read_sn}),
then $rop$ is classified as a silent \rd{}, and it appears in $L(\alpha)$,
by rule $R1$ if $R.seq = x$ when $SN$ is read or rule $R2$, otherwise.
If (2) holds ($p_j$ applies a \fx{} to $R$), then $op$ is a direct \rd{}, 
and linearized in $L(\alpha)$ by rule $R1$.

If neither (1) nor (2) hold, then $p_j$ has either not read $SN$, 
or read a value $\neq prev\_val$ from $SN$ but without yet accessing $R$. 
In both cases, $op$ is unclassified and hence not linearized. 
\end{proof}

We can prove that an audit $aop$ includes a pair $(j,v)$ in its response set
\emph{if and only if} 
a \rd{} operation by process $p_j$ with output $v$ is linearized before it.
Since a \rd{} is linearized if and only it is effective 
(Lemma~\ref{lem:eff_read_lin}), 
any \adt{} operation that is linearized after 
the \rd{} is effective, must report it. 
This implies:

\begin{lemma}
\label{lem:lin_adt_after_eff_rd}
If an \adt{} operation $aop$ is invoked and returns 
in an extension $\alpha'$ of $\alpha$,
and $\alpha$ contains a $v$-effective \rd{} operation by process $p_j$, 
then $(j,v)$ is contained in the response set of $aop$.
\end{lemma}



Lemma~\ref{lem:writeNonAudit} shows that writes are uncompromised by readers, 
namely, a read cannot learn of a value written, 
unless it has an effective \rd{} that returned this value.
Lemma~\ref{lem:rd_non_auditable} shows that reads are uncompromised by other readers, 
namely, they do not learn of each other.

\begin{restatable}{lemma}{writeNonAudit}
\label{lem:writeNonAudit}
Assume $p_j$ only performs \rd{} operations.
Then for every value $v$ either there is a \rd{} operation by $p_j$ in $\alpha$ that is $v$-effective,
or there is $\alpha'$, $\ind{\alpha'}{p_j}{\alpha}$ in which no \wrt{} has input $v$.
\end{restatable}

\begin{proof}
If $v$ is not an input of some \wrt{} operation in $\alpha$, the lemma follows by taking $\alpha' = \alpha$. 
If there is no visible \wrt{}($v$) operation in $\alpha$,
then, since a silent \wrt{}($v$) does not change $R.val$ to $v$, 
the lemma follows by changing its input to some value $v' \neq v$ to obtain 
an execution $\ind{\alpha'}{p_j}{\alpha}$

Let $wop$ be a visible \wrt{}($v$) operation in $\alpha$. 
Since it is visible, $wop$ applies a \cas{} to $R$ that changes 
$(R.seq,R.val)$ to $(x,v)$ where $x$ is some sequence number. 
If $p_j$ applies a \fx{} to $R$ while $R.val = v$, 
then the corresponding \rd{} operation $rop$ it is performing is direct and $v$-effective. 
Otherwise, $p_j$ never applies a \fx{} to $R$ while $R.val = v$. 
$R$ is the only shared variable in which inputs of \wrt{} are written 
and that is read by $p_j$.
Hence, the input of $wop$ can be replaced by another value $v'\neq v$, 
creating an indistinguishable execution $\alpha'$ 
without a \wrt{} with input $v$. 
\end{proof}

\begin{restatable}{lemma}{uncompromisedR}
  \label{lem:rd_non_auditable}
Assume $p_j$ only performs \rd{} operations, 
then for any reader $p_k$, $k \neq j$, there is an execution 
$\ind{\alpha'}{p_j}{\alpha}$ in which no \rd{} by $p_k$ is $v$-effective,
for any value $v$.
\end{restatable}

\begin{proof}
The lemma clearly holdes if there is no $v$-effective \rd{} by process $p_k$. 
So, assume there is a $v$-effective \rd{} operation $rop$ by $p_k$. 
Let $\alpha'$ be the execution in which we remove all $v$-effective \rd{} 
operations performed by $p_k$ that are silent. 
Such operations do not change any shared variables, 
and therefore, $\ind{\alpha'}{p_j}{\alpha}$. 

So, let $rop$ be a direct, $v$-effective \rd{} by $p_k$. 
When performing $rop$, $p_k$ applies \fx{} to $R$ 
(line~\ref{line:mwmrma:read_fetch}), when $(R.seq,R.val) = (x,v)$, 
for some sequence number $x$. 
This step only changes the $k$th tracking bit of $R$ unchanged to, say, $b$. 
Recall that $R$ is accessed (by applying a \fx{}) at most once by $p_j$ while $R.seq = x$. 
If no \fx{} by $p_j$ is applied to $R$ while $R.seq = x$, or one is applied before $p_k$'s, 
$rop$ can be removed without being noticed by $p_j$. 
Suppose that both $p_k$ and $p_j$ apply a \fx{} to $R$ 
while $R.seq = x$, and that $p_j$'s \fx{} is after $p_k$'s. 
Let $\alpha'_{x,b}$ be the execution identical to $\alpha'$, 
except that (1) the $k$th bit of $rand_x$ is $b$ and, 
(2) $rop$ is removed. 
Therefore, $\ind{\alpha'_{x,b}}{p_j}{\alpha'}$,
and since $\ind{\alpha'}{p_j}{\alpha}$,
we have that $\ind{\alpha'_{x,b}}{p_j}{\alpha}$. 
\end{proof}


\begin{theorem}
Alg.~\ref{alg:cas_mwa} is a linearizable and wait-free implementation of 
an auditable multi-writer, multi-reader register.
Moreover, 
  \begin{itemize}
  \item An \adt{} reports $(j,v)$ 
  if and only if $p_j$ has an $v$-effective \rd{} operation in $\alpha$. 
  \item a \wrt{$v$} is uncompromised by a reader $p_j$, unless $p_j$ has a $v$-effective \rd{}.
  \item a \rd{} by $p_k$ is uncompromised by a reader $p_j \neq p_k$.
  \end{itemize}
\end{theorem}

\section{An Auditable Max Register}
\label{sec:maxreg}

This section shows how to extend the register implementation of the previous section
into an implementation of a max register with the same properties.
A \emph{max register} provides two operations:
$\wrtm{}(v)$ which writes a value $v$ and \rd{} which returns a value. 
Its sequential specification is that a \rd{} returns 
the largest value previously written. 
An auditable max register also provides an \adt{} operation, 
which returns a set of pairs $(j,v)$. 
As in the previous section, 
reads are audited if and only if they are effective, 
and readers cannot compromise other $\wrtm{}$ operations, unless they read them,
or other \rd{} operations. 



Alg.~\ref{alg:maxreg} uses essentially the same \rd{} and \adt{} as in Alg.~\ref{alg:cas_mwa}. 
The \wrtm{} operation is also quite similar, with the following differences (lines in blue in the pseudo-code).
In Alg.~\ref{alg:cas_mwa}, a \wrt{}($w$) obtains a new sequence number $s+1$
and then attempts to  change $R$ to $(s+1,w,rand_{s+1})$. 
The operation terminates after it succeeds in doing so, 
or if it sees in $R$ a sequence number $s' \geq s+1$. 
In the latter case, a concurrent \wrt{}($w'$) has succeeded and may be seen as overwriting $w$,
so \wrt{}($w$) can terminate, even if $w$ is never written to $R$.
The implementation of  \wrtm{}  uses a similar idea, except that 
(1) we make sure that the successive values in $R$ are non-decreasing and 
(2) a \wrtm{}{}($w$) with sequence number $s+1$ is no longer abandoned when a sequence number 
$s'\geq s+1$ is read from $R$, but instead when $R$ stores a value $w' \geq w$. 


There is however, a subtlety that must be taken care of.
A reader may obtain a value $v$ with sequence number $s$, 
and later read a value $v+2$ with sequence number $s' > s+1$.
This leaks to the reader that some \wrtm{} operations occur in between 
its \rd{} operations, and in particular, that a $\wrtm{}(v+1)$ occurred, 
without ever effectively reading $v+1$.

To deal with this problem, we append a \emph{random nonce} $N$ 
to the argument of a \wrtm{} operation, where $N$ is a random number.
The pair $(w,N)$ is used as the value written $v$ was used in Alg.~\ref{alg:cas_mwa}.
The pairs $(w,N)$ are ordered lexicographically, that is, 
first by their value $w$ and then by their nonce $N$. 
Thus, the reader cannot guess intermediate values. 
The code for \rd{} and \adt{} is slightly adjusted in 
Alg.~\ref{alg:maxreg} 
versus Alg.~\ref{alg:cas_mwa}, 
to ignore the random nonce $N$ from the pairs when values are returned.


\begin{algorithm}[tb]
  \begin{algorithmic}[1]
    \setcounter{ALG@line}{20}
    \Statex \textbf{shared registers}
    \Statex\hspace{\algorithmicindent} $R,\mathit{SN}, V[0..+\infty], B[0..+\infty][0..m-1]$ as in Alg.~\ref{alg:cas_mwa}
    \Statex\hspace{\algorithmicindent} \ha{$M$: a (non-auditable) max register, initially $v_0 = (w_{0},N_0)$}
    \Statex\textbf{local variables: writer, reader, auditor, } as in Alg.~\ref{alg:cas_mwa}
  \Function{read($\:$), audit}{$\:$}: same as in Alg~\ref{alg:cas_mwa}
  \EndFunction

  \Function{writeMax}{$w$}
    \State \hspace{-3pt}\ha{$v \gets (w,N)$}, where $N$ is a  fresh random nonce
  \State \hspace{-3pt}\ha{$M.\mathsf{writeMax}(v)$}; $sn \gets \mathit{SN}.\mathsf{read}() +1$; \label{lmr:writer_read1_SN} \label{lmr:writer_write_M}
  \Repeat\label{lmr:writer_repeat}
  \State  $(lsn,lval,bits) \gets R.\mathsf{read}()$ \label{lmr:writer_read_R}
  \State \ha{\textbf{if} $lval \geq v$ \textbf{then} $sn \gets lsn$; \textbf{break}} \label{lmr:writer_break}
  \State \ha{\textbf{if} $lsn \geq sn$ \textbf{then}}
  \State \hspace{6pt}\ha{$\mathit{SN}.\mathsf{compare\&swap}(sn-1,sn)$;} 
  \State \hspace{6pt}\ha{$sn \gets \mathit{SN}.\mathsf{read}() +1$; \label{lmr:writer_read2_SN} \textbf{continue}} \label{lmr:writer_sn_used} 
  \State \ha{$mval \gets M.\mathsf{read}()$}\label{lmr:writer_read_M}
  \State $V[lsn].\mathsf{write}(lval.value)$; 
  \State $B[lsn][j].\mathsf{write}(\emph{true})$ $\forall j$, s.t. $bits[j] \neq rand_{lsn}[j]$ \label{lmr:writer_record_reads}
  \Until{$R.\mathsf{compare\&swap}((lsn,lval,bits),(sn,\ha{mval},rand_{sn}))$}\label{lmr:writer_cas_R}
   \State $\mathit{SN}.\mathsf{compare\&swap}(sn-1,sn)$; \Return{}\label{lmr:writer_return}
  \EndFunction
\end{algorithmic}
\caption{Auditable Max Register}
\label{alg:maxreg}
\end{algorithm}

In the algorithm, a (non-auditable) max-register $M$ is shared among the writers. 
A \wrtm($w$) by $p$ starts by writing the pair $v=(w,N)$ of the value $w$ and the nonce $N$ to $M$, 
before entering a repeat loop. 
Each iteration is an attempt to store in $R$ the current value $mval$ of $M$, 
and the loop terminates as soon as $R$ holds a value equal to or larger than $mval$. 
Like in Alg.~\ref{alg:cas_mwa}, $R$ holds a triplet $(s,val,bits)$ 
where $s$ is $val$’s sequence number, $val$ is the current value, 
and $bits$ is the encrypted set of readers of $val$. 
Before attempting to change $R$, $val$ and the set of readers, 
once deciphered, are stored in the registers $V[s]$ and $B[s]$, 
from which they can be retrieved with \adt{}.

In each iteration of the repeat loop, the access pattern of \wrt{} in 
Alg.~\ref{alg:cas_mwa} to the shared register $\emph{SN}$ and $R$ is preserved. 
After obtaining a new sequence number $s+1$, 
where $s$ is the current value of $SN$ (line~\ref{lmr:writer_read1_SN} 
for the first iteration, line~\ref{lmr:writer_read2_SN} otherwise), 
a triple $(lsn,lval,bits)$ is read from $R$. 
If $lval \geq v$, the loop breaks as a value that is equal to or larger than $v$ 
has already been written. 
As in Alg.~\ref{alg:cas_mwa}, before returning  
we make sure that the sequence number in $SN$ is at least as large as $lsn$, 
the sequence number in $R$.


\section{Auditable Snapshot Objects and Versioned Types}
\label{sec:snapshot}

We show how an auditable max register (Section~\ref{sec:maxreg})
can be used to make other object types auditable.


\subsection{Making Snapshots Auditable}

We start by showing how to implement an auditable $n$-component snapshot object,
relying on an auditable max register. 
Each component has a state, initially $\bot$, and a different designated writer process. 
A \emph{view} is an $n$-component array, each cell holding a value written by a process in its component. 
A \emph{atomic object} \cite{AfekADGMS93} provides two operations: $\upd(v)$ that changes the process's component to $v$, and \scn{} 
that returns a view. 
It is required that in any sequential execution, in  the view returned by a \scn{}, each component contains the value of the latest \upd{} to this component (or $\bot$ if there is no previous \upd{}). 
As for the auditable register, an \adt{} operation returns a set of pairs $(j,view)$. 
In a sequential execution, there is  such a pair if and only if  the operation is preceded by a \scn{} by process $p_j$ that returns $view$. 
Here, we want that audits report exactly those \scn{s} that have made enough progress to infer the current $view$ of the object.

Denysuk and Woeffel~\cite{DenysyukW15} show that a strongly-linearizable max register 
can be used to transform a linearizable snapshot into its strongly linearizable counterpart. 
As we explain next, with the same technique, non-auditable snapshot objects can be made auditable. 
Algorithm~\ref{alg:snapshot} adds an \adt{} operation to their algorithm.
Their implementation is lock-free, 
as they rely on a lock-free implementation of a max register. 
Algorithm~\ref{alg:snapshot} is wait-free since we use the
\emph{wait-free} max-register implementation of Section~\ref{sec:maxreg}. 

Let $S$ be a linearizable, but non-auditable snapshot object.
The algorithm works as follows: each new state (that is, whenever one component is  updated) is associated with a unique and increasing \emph{version number}.  
The version number is obtained by storing a sequence number $sn_i$ 
in each component $i$ of $S$, in addition to its current value.
Sequence number $sn_i$ is incremented each time the $i$th component is updated 
(line~\ref{lsp:upd_write_S}).  
Summing the sequence numbers of the components yields 
a unique and increasing version number ($vn$) for the current view. 

The pairs $(vn,view)$, where $vn$ is a version number and $view$ 
a state of the  auditable snapshot, 
are written to an auditable max register $M$.
The pairs are ordered according to the version number,
which is a total order since version numbers are unique. 
Therefore, the latest state can be retrieved by reading $M$, 
and the set of past \scn{} operations can be obtained by auditing $M$ 
(line~\ref{lsp:adt_M}). 
The current view of the auditable snapshot is stored in $S$. 

In an \upd{}($v$), process $p_i$ starts by updating the $i$th 
component of $S$ with $v$ and incrementing the sequence number field $sn_i$. 
It then scans $S$, thus obtaining a new  view of $S$ that includes its update.
The view $view$ of the implemented auditable snapshot is obtained by removing the sequence number 
in each component (line~\ref{lsp:upd_new_view}). 
The version number $vn$ associated with this view is 
the sum of the sequence numbers. 
It then writes $(vn,view$) to the max-register $M$ (line~\ref{lsp:upd_return}). 
A \scn{} operation reads a pair $(vn,view)$ from $M$ and returns 
the corresponding $view$ (line~\ref{lsp:snp_return}). 
Since $M$ is auditable, the views returned by the processes that 
have previously performed a \scn{} can thus be inferred by auditing $M$ 
(line~\ref{lsp:adt_M}).

\begin{algorithm}[tbp]
  \caption{$n$-component auditable snapshot objects.}
  \label{alg:snapshot}
  \begin{algorithmic}[1]
    \Statex\textbf{shared registers}
    \Statex\hspace{\algorithmicindent} $\mathit{M}$: auditable  max register, initially $(0,[\bot,\ldots,\bot])$
    \Statex\hspace{\algorithmicindent} $\mathit{S}$: (non-auditable) snapshot object, 
    \Statex\hspace{\algorithmicindent}\hspace{\algorithmicindent} initially $[(0,\bot),\ldots,(0,\bot)]$  
    \Statex\textbf{local variable: writer} $p_i, 1 \leq i \leq n$
    \Statex\hspace{\algorithmicindent} $sn_i$ local sequence number, initially $0$ 
    \Function{update}{$v$} \Comment{code for writer $p_i, i \in \{1,\ldots,n\}$}
    \State $sn_i \gets sn_i + 1$; $S.\mathsf{update}_i((sn_i,v))$\label{lsp:upd_write_S}
    \State $sview \gets S.\mathsf{scan}()$;  $vn \gets \sum_{1 \leq j \leq n} sview[j].sn$\label{lsp:upd_scan_S}
    \State $view \gets $ the $n$-component array of the values in $sview$ \label{lsp:upd_new_view}
    \State $M.\mathsf{writeMax}((vn,view))$; \Return\label{lsp:upd_return}
    \EndFunction

    \Function{scan}{$\,$} 
      \State $(\_,view) \gets M.\mathsf{read}()$; \Return $view$\label{lsp:snp_return}
   \EndFunction

   \Function{audit}{$\,$} 
     \State $\mathit{MA} \gets M.\mathsf{audit}()$; 
     \State \Return $\{(j,view) : \exists ~\mbox{an element}~ (j,(*,view)) \in \mathit{MA}\}$\label{lsp:adt_return} \label{lsp:adt_M} 
     \EndFunction
    \end{algorithmic}
\end{algorithm}

The \adt{} and \scn{} operations interact with the implementation by applying 
a single operation (audit and read, respectively) to the auditable max register $M$.
The algorithm therefore lifts the properties of the implementation of $M$ 
to the auditable snapshot object.
In particular, when the implementation presented in Section~\ref{sec:maxreg} is used, 
effective \scn{} operations are auditable, 
\scn{} operations are uncompromised by other scanners, and
\upd{} operations are uncompromised by scanners. 

\subsection{Proof of Correctness}
\label{app:proof of snapshot}

Let $\alpha$ be a finite execution of Algorithm~\ref{alg:snapshot}. 
To simplify the proof, we assume the inputs of \upd{} by the same process are unique. 

We assume that the implementation of $M$ is wait-free and linearizable.
In addition, it guarantees effective linearizability and that \rd{} operations 
are uncompromised by other readers. 
We also assume that the implementation of $S$ is linearizable and wait-free
(e.g.,\cite{AfekADGMS93}).
Inspection of the code shows that \upd, \scn{} and \adt{} operations are wait-free. 

Since $S$ and $M$ are linearizable and linearizability is composable, 
$\alpha$ can be seen as a sequence of steps applied to $S$ or $M$.
In particular, we associate with each high-level operation $op$ a step $\sigma(op)$ 
applied by $op$ either to $S$ or to $M$. 
The linearization $L(\alpha)$ of $\alpha$ is the sequence formed by ordering the operations 
according to the order their associated step occurs in $\alpha$. 

For a \scn{} and an \adt{} operation $op$, $\sigma(op)$ is, respectively, 
the $\mathsf{read}$ and the $\mathsf{audit}$ steps applied to $M$. 
If $op$ is an \upd{} with input $x$ by process $p_i$, 
then let $vn_x$ be the sum of the sequence numbers $sn$ in each component of $S$ 
after $\emph{update}(x)$ has been applied to $S$ by $p_i$. 
$\sigma(op)$ is the first $\mathsf{write}$ to $M$ of a pair $(vn,view)$ 
with $vn \geq vn_x$ and $view[i] = x$. 
If there is no such $\mathsf{write}$, $op$ is discarded. 

We first show that the linearization $L(\alpha)$ respects 
the real-time order between operations.

  \begin{lemma} 
  \label{lem:Snapshot-real-time-o}
  If an operation $op$ completes before an operation $op'$ is invoked in $\alpha$, then $op$ precedes $op'$ in $L(\alpha)$. 
  \end{lemma}
  \begin{proof}
  We show that that the linearization point of any operation $op$ is inside its execution interval;
  the claim is trivial for \scn{} or \adt{} operations.

  Suppose that $op$ is an \upd{} by a process $p_i$ with input $x$. 
  The sum of the sequence numbers in the components of $S$ increases each time an $\mathsf{update}$ is applied to it. 
  Hence, any pair $(vn,view)$ written to $M$ before $p_i$ has updated its component of $S$ to $x$ 
  is such that $vn < vn_x$. Therefore $\sigma(op)$, if it exists, is after $op$ starts. 
  If $op$ terminates, then it scans $S$ after updating the $i$th component of $S$ to $x$. 
  The $view$ it obtains and its associated version number satisfy $view[i] = x$ and $vn \geq vn_x$. 
  This pair is written to $M$. 
  If $\sigma(op)$ is not this step, then $\sigma(op)$ occurs before $op$ terminates. 
  If $op$ does not terminate and $\sigma(op)$ does exist, 
  it occurs after $op$ starts and thus within $op$'s execution interval. 
\end{proof}

\begin{lemma}\label{lem:snap-sequential-spec} 
  Each component $i$ of the view  returned by a \scn{} is the input of the last \upd{} by $p_i$ linearized before the \scn{}
  in $L(\alpha)$.
 \end{lemma}  
 \begin{proof}
  Consider a \scn{} operation $sop$ that returns  $view$, with $view[i] = x$. This view is read from the max register $M$ and has version number $vn$. Let $op$ be the last \upd{} by $p_i$ linearized before $sop$ in $L(\alpha)$, let $y$ be its input and $vn_y$ the version number (that is the sum of the sequence number stored in each component) of $S$ immediately after $S.\mathsf{update}(y)$ is applied by $p_i$. 
  
  We denote by $\sigma_u$ this low level $\mathsf{update}$. 
  Since the version number increases with each $\mathsf{update}$, 
  every pair $(vn',view')$ written into $M$ before $\sigma_u$ is such that $vn' < vn_y$.
  Also, every pair $(vn',view')$ written to $M$ after $\sigma_u$ and before $sop$ is linearized satisfies $vn' \geq vn_y \implies view'[i] = y$. Indeed, if $vn' \geq vn_y$, $view'$ is obtained by a $\mathsf{scan}$ of  $S$ applied after the $i$-th component is set to $y$. Hence, $view'[i] = y$ 
  because we assume that $op$ is the last \upd{} by $p_i$ linearized before $sop$ in $L(\alpha)$. 
  
  Finally, step  $\sigma(op)$ is a $\mathsf{write}$ of pair $(vn',view')$ to $M$ with $vn' \geq vn_y$ and $view'[i] =y$. $\sigma(op)$ occurs after $\sigma_u$ and before the max register $M$ is read by $sop$. It thus follows that the pair $(vn,view)$ read from $M$ in $sop$ satisfies $vn \geq vn_y$ and has been written after $\sigma_y$. Hence, $view[i] = y = x$. We conclude that each component $i$ of the view  returned by a \scn{} is the input of the last \upd{} by $p_i$ linearized before the \scn{}
  in $L(\alpha)$.
\end{proof}

\begin{lemma}
  \label{lem:snap_audit}
  An \adt{} reports $(j,view)$ if and only if $p_j$ has a $view$-effective\footnote{
    Namely, $p_i$ has a \scn{} operation that returns $view$ in all indistinguishable executions.} 
    \scn{} in $\alpha$.
  Each $\upd{}(v)$ is uncompromised by a scanner $p_j$ unless it has a $view$-effective \scn{} with one component of $view$ equal to $v$. Each \scn{} by $p_k$ is uncompromised by a scanner $p_j \neq p_k$. 
\end{lemma}
\begin{proof}
  A \scn{} applies a single operation on  shared objects, namely a $\mathsf{read}$ on $M$. It is linearized with this step, which determines  the view it returns. Therefore, a \scn{} is linearized if and only if it is effective. Hence $(j,view$) is reported by an \adt{} if and only if $p_j$ has a $view$-effective \scn{}.

  Let $v$ be the input of an \upd{} operation by some process $p_i$. If there is no $view$ with $view[i] = v$ written to $M$ (line 5), \upd{}($v$) can be replaced by \upd{}($v'$), $v' \neq v$ in an execution $\alpha'$, $\ind{\alpha}{p_j}{\alpha'}$. Otherwise, note that each $sview$ for which  $p_j$ has a $sview$-effective \scn{}, we have  $sview[i] \neq v$. Suppose that  $view$, with $view[i] = v$ is written to $M$ in $\alpha$. Then we can replace $view$ with an array $view'$, identical to $view$ except that $view'[i] = v' \neq v$ an execution $\ind{\alpha'}{p_j}{\alpha}$. This is because the write of $view$  is not compromised by $p_j$ in $M$. By repeating this procedure until all writes to $M$  of $view$s with $view[i] = v$ have been eliminated leads to an execution $\beta, \ind{\beta}{p_j}{\alpha}$ in which there is no \upd{}($v$). 
\end{proof}

\begin{theorem}
  \label{thm:audtible_snapshot}
  Alg.~\ref{alg:snapshot} is a wait-free linearizable implementation 
  of an auditable snapshot object which audits effective \scn{} operations,
  in which \scn{} and \upd{} are uncompromised by scanners.
\end{theorem}

\subsection{Versioned Objects}

Snapshot objects are an example of a \emph{versioned type}~\cite{DenysyukW15},
whose successive states are associated with unique and increasing version numbers. 
Furthermore, the version number can be obtained from the object itself, 
without resorting to external synchronization primitives. 
Essentially the same construction can be applied to any versioned object.

An object $t \in \m{T}$ is specified by a tuple  $(Q,q_0,I,O,f,g)$, 
where $Q$ is the state space,  
$I$ and $O$ are respectively the input and output sets of $\mathsf{update}$ and $\mathsf{read}$ operations. 
$q_0$ is the initial state and functions $f: Q \to O$ and 
$g: I\times{}Q \to Q$ describes the sequential behavior of $\mathsf{read}$ and $\mathsf{update}$. 
A $\mathsf{read}()$ operation leaves the current state $q$ unmodified and returns $f(q)$. 
An $\mathsf{update}(v)$, where $v \in I$ changes 
the state $q$ to $g(v,q)$ and does not return anything.

A linearizable \emph{versioned} implementation of a type $t \in \m{T}$ 
can be transformed into a strongly-linearizable one~\cite{DenysyukW15}, 
as follows. 
Let $t = (Q,q_0,I,O,f,g)$ be some type in $\m{T}$. 
Its \emph{versioned} variant $t' = (Q',q_0',I',O',f',g')$ has 
$Q' = Q \times \mathbb{N}$, $q_0' = (q_0,0)$, $I' = I$, $O' = O \times \mathbb{N}, f':Q' \to O\times \mathbb{N}$ and $g': I \times{}Q' \to Q'$. 
That is, the state of $t'$  is augmented  with a version number, 
which increases with each $\mathsf{update}$ and is 
returned by each $\mathsf{read}$: $f'((q,vn)) = (f(q),vn)$ and $g'((q,vn)) = (g(q),vn')$ with $vn < vn'$. 

A versioned implementation of a type $t \in \m{T}$ can be transformed into 
an auditable implementation of the same type using an auditable register. 
The construction is essentially the same as presented in Algorithm~\ref{alg:snapshot}.
In the auditable variant $T_a$ of $T$, to perform an \upd($v$), 
a process $p$ first update the versioned implementation $T$ before reading it. 
$p$ hence obtains a pair $(o,vn)$ that it writes to the auditable max register $M$. 
For a \rd{}, a process returns what it reads from $M$. 
As \rd{} amounts to read $M$, to perform an \adt{} a process simply audit the max-register $M$. 
As we have seen for snapshots, 
$T_a$ is linearizable and wait-free. 
Moreover, $T_a$ inherits the advanced properties of the underlying max-register:
If $M$ is implemented with Algorithm~\ref{alg:maxreg}, 
then it correctly audits effective \rd{}, and \rd{} and \upd{} are uncompromised.

\begin{theorem}[versioned types are auditable]
  \label{thm:auditable_versioned}
  Let  $t \in \m{T}$, and let $T$ be a  versioned implementation of $t$ that is linearizable and wait-free. There exists a wait-free, linearizable and auditable implementation of $t$ from $T$ and auditable max-registers in which \rd{} and \upd{} are uncompromised by readers  and \adt{} reports only effective \rd{} operations. 
\end{theorem}


\section{Discussion}

This paper introduces novel notions of auditability that deal 
with curious readers. 
We implement a wait-free linearizable auditable register 
that tracks effective reads 
while preventing unauthorized audits by readers.
This implementation is extended into an auditable max register,
which is then used to implement auditable atomic snapshots and versioned types. 

Many open questions remain for future research.
An immediate question is how to implement an auditable register in which
\emph{only auditors can audit}, i.e., reads are uncompromised by writers. 
%
A second open question is how to extend auditing to additional objects.
These can include, for example, \emph{partial snapshots}~\cite{AttiyaGR08} in which a reader 
can obtain an ``instantaneous'' view of a subset of the components.
Another interesting object is a \emph{clickable} atomic 
snapshot~\cite{JayantiJJ24}, 
in particular, variants that allow arbitrary operations on the 
components and not just simple updates (writes).

The property of uncompromising other accesses can be seen as 
an \emph{internal} analog of \emph{history independence}, 
recently investigated for concurrent objects~\cite{AttiyaBFOS24}.
A history-independent object does not allow an external observer, \emph{having access 
to the complete system state}, to learn anything about operations applied to the object,
but only its current state. 
Our definition, on the other hand, does not allow an internal observer, 
e.g., a reader that only reads shared base objects, 
to learn about other \rd{} and \wrt{} operations applied in the past. 
An interesting intermediate concept would allow several readers \emph{collude} 
and to combine the information they obtain in order to learn more than what they are allowed to. 



\begin{acks}
H. Attiya is supported by the Israel Science Foundation (grant number 22/1425). 
A. Fernández Anta has been funded by project PID2022-140560OB-I00 (DRONAC) funded by MICIU / AEI / 10.13039 / 501100011033 and ERDF, EU. 
A. Milani is supported by the France 2030 ANR project ANR-23-PECL-0009 TRUSTINCloudS. 
C. Travers is supported in part by ANR projects DUCAT (ANR-20-CE48-0006). 
\end{acks}

\bibliographystyle{ACM-Reference-Format}
\bibliography{refs}

\appendix
\section{Additional Proofs for Algorithm~\ref{alg:cas_mwa} (Auditable Register)}
\label{app:proof of register}

Simple code inspection (line~\ref{line:mwmrma:reader_cas}, 
line~\ref{line:mwmrma:writer_cas}, and line~\ref{line:mwmrma:audit_return})
shows:

\begin{invariant}
\label{obs:R_S_increase 1}
The successive values of $SN$ are $0,1,2,\ldots$.
\end{invariant}

\begin{restatable}{invariant}{InvTwo}
\label{obs:R_S_increase 2}
The successive values of $R.seq$ are strictly increasing. 
\end{restatable}


\begin{proof}
The proof is by induction on the length of the execution;
the invariant clearly holds for an empty execution.
Consider a step that changes $R.seq$ to $x$,
which only happens when a successful \cas{} is applied by some process $p$,
in line~\ref{line:mwmrma:writer_cas}.
Before this step is applied, $p$ reads $R$ (line~\ref{line:mwmrma:writer_read_sn}) to make sure that $R.seq$ is strictly smaller than $x$ (otherwise, the repeat loop terminates without applying a \cas{} to $R$ (line~\ref{line:mwmrma:writer_break})). 
If the \cas{} of line~\ref{line:mwmrma:writer_cas} succeeds, $R.seq$ has not been modified since it was last read by $p$, and its value increases to $x$. 
\end{proof}


\begin{restatable}{lemma}{LemmaSNR}
\label{lem:SN_R}
After any finite execution $\alpha$, 
and for any integer $x \geq 0$, (1) $SN = x \implies R.seq \geq x$
and, (2) $R.seq = x \implies SN \geq x-1$.
\end{restatable}


\begin{proof}
The proof is by induction on the length of $\alpha$, 
and both claims trivially hold after the empty execution, 
since $R.seq = SN = 0$ in the initial configuration. 
Assume that both claims hold after a finite prefix $\alpha$,
and consider the first step that modifies $R.seq$ or $\mathit{SN}$. 

If the step modifies $\mathit{SN}$, then it is 
a successful \cas{} applied by some process $p$ when performing 
a \rd{} (line~\ref{line:mwmrma:reader_cas}), 
a \wrt{} (line~\ref{line:mwmrma:writer_return}), or 
an \adt{} (line~\ref{line:mwmrma:audit_return}). 
Let $x$ be the new value of $SN$ after the \cas{} is applied. 
If $p$ is performing a \rd{} or an \adt{}, 
$p$ has previously read $x$ from $R.seq$ (line~\ref{line:mwmrma:read_fetch} or line~\ref{line:mwmrma:audit_read}). 
Since the values of $R.seq$ do not decrease, 
$R.seq \geq x$ after the successful \cas{} applied to $SN$ by $p$.
If $p$ is performing a \wrt{}, it has previously read $x$ from $R.seq$ (line~\ref{line:mwmrma:writer_read_R}) or has changed its value to $x$ (line~\ref{line:mwmrma:writer_cas}) by applying a successful \cas{}. 
Since successive values of $R.seq$ are increasing (Invariant~\ref{obs:R_S_increase 2}), $R.seq \geq x$ after $p$ changes $SN$ to $x$. 
(2) also holds since it holds before this step, 
and continues to hold because the value of $SN$ increases. 

If the step sets $R.seq$ to $x$, then it is 
a successful \cas{} applied to $R$ by some process $p$ 
while performing a \wrt{} operation  (line~\ref{line:mwmrma:writer_cas}). 
Before applying this \cas{}, $p$ reads $x-1$ from $SN$ (line~\ref{line:mwmrma:writer_read_sn}). 
Since successive values of $SN$ are increasing (Invariant~\ref{obs:R_S_increase 1}), 
$SN \geq x-1$ after this step. 
(1) also holds after this step since $R.seq$ 
is changed to a value larger than $x$.
\end{proof}

\begin{restatable}{lemma}{ReadOnce}
  \label{lem:reader_once}
  Let $\sigma,\sigma'$ be two \fx{} applied to $R$ by the same reader $p$. Let $(sr,vr,br)$ and $(sr',vr',br')$ be the values of $R$ immediately before these steps are applied, respectively, then $sr \neq sr'$. 
\end{restatable}

\begin{proof}
  Suppose that $\sigma$ is applied before $\sigma'$. By the code, they are applied when $p$ is performing two distinct \rd{} operations denoted  $rop$ and $rop'$ respectively. 
  By line~\ref{line:mwmrma:read_fetch}, after $\sigma$, the value of the local variable $sn$ at process $p$ is $sr$ and by Lemma~\ref{lem:SN_R}, $SN \geq sr - 1$.  Then, $p$ applies a \cas{} (line~\ref{line:mwmrma:reader_cas}) with parameter $(sr-1,sr)$. After this step, the value of  $SN$ is thus  $\geq sr$, as successive value of $SN$ are increasing. 
  Note also that the local variable  $prev\_sn$ is set to $sr$. 

  In $rop'$, $p$  reads from $SN$ (line~\ref{line:mwmrma:read_sn}) a value $sn' > sr$. Otherwise, $sn' = prev\_sn$ and no \fx{} is applied to $R$. 
  But this means by Lemma~\ref{lem:SN_R}  that the sequence number stored in $R$ is also strictly greater than $sr$. 
  Therefore, as the successive sequence number stored in $R$ are increasing, $sr' > sr$. 
\end{proof}

The next lemma shows that every value is associated 
with a unique sequence number in $R$. 

\begin{restatable}{lemma}{LemmaR}
  \label{lem:R}
  Let $\alpha$ be a finite execution. There exists $k \geq 0$ and inputs  of \wrt{} operations $v_1,\ldots,v_k$ such that the sequence of values of the first two fields $(R.seq, R.val)$ of $R$  is   $(0,v_0), (1,v_1), \ldots, (k,v_k)$. 
\end{restatable}


\begin{proof}
  Note that the initial value of $R$ is $(0,v_0)$. Suppose that there exists inputs of \wrt{} operations $v_1,\ldots,v_\ell$ such that the first $\ell+1$ values of the couple $(R.seq,R.val)$ are
  $(0,v_0), \ldots, (\ell,v_\ell)$. If $(R.seq,R.val)$ no longer changes after it is set to  $(\ell,v_\ell)$, the Lemma is true. Otherwise, let us consider the first step $\sigma$ that changes $R$ from $(\ell,v_\ell,b)$ to some triple $(\ell',v',b')$ with $(\ell,v_\ell) \neq (\ell',v')$. This step is a successful  \cas{} applied during a \wrt{} whose input is $v'$ by some process $p$ at line~\ref{line:mwmrma:writer_cas}, since this is the only place in which  $R.seq$ or $R.val$ is changed (each \fx{} applied to $R$ by a reader changes only one of the last $m$ bits of $R$, leaving the first two fields unmodified). By the code, at the beginning of this  \wrt{}, $p$ reads $\ell'-1$ from $SN$  (line~\ref{line:mwmrma:writer_read_sn}). Hence, it follows from Lemma~\ref{lem:SN_R} that immediately after this read, $\ell'-1 \leq R.seq$. 

  Before  $(R.seq,R.val)$ is changed to $(\ell',v')$, $R$ is read (line~\ref{line:mwmrma:writer_read_R}). The triple returned is $(\ell,v_\ell,b)$, since otherwise the \cas{} is not successful as it is the first step in which   $(R.seq,R.val)$ changes from  $(\ell,v_\ell)$ to a different value. Note that $\ell < \ell'$, since otherwise $p$ exits the repeat loop without applying a \cas{} to $R$ (line~\ref{line:mwmrma:writer_break}). 

  Moreover, as the read of $SN$ (line~\ref{line:mwmrma:writer_read_sn}, after which we have $\ell'-1 \leq R.seq$) occurs before $(\ell,v_\ell,b)$ is read from $R$, and as sequence numbers in $R$ are increasing (Invariant~\ref{obs:R_S_increase 2}), then
  $\ell'-1 \leq \ell$. Hence $\ell'-1 \leq \ell < \ell'$, from which we conclude that $\ell' = \ell+1$. Therefore,  after step $\sigma$, the new value of $(R.seq,R.val)$ is  $(\ell+1,v_{\ell+1})$, where  $v_{\ell+1} = v'$ is the input of $p$'s \wrt{} operation. 
\end{proof}

\begin{lemma}
  \label{lem:unique_visible}
  Let $x \geq 0$ such that there is in $H'$ a \rd{} or \wrt{} operation associated with sequence number $x$. There exists a unique visible \wrt{} operation $wop$ with sequence number $sn(wop) = x$. 
\end{lemma}
\begin{proof}
  The lemma is true for $x=0$. For $x > 1$, let us first suppose that there exists a silent \wrt{} operation $op$ by some process $p$ with $sn(op) = x$. As $op$ is silent, $p$ reads from $R$ a sequence number $lsn \geq x$ at line~\ref{line:mwmrma:writer_break}. It follows from Lemma~\ref{lem:R} that before this read, the  field $R.seq$ of $R$  has been set to $x$. If $op$ is a silent \rd{} operation, $sn(op)$ is the  sequence number read from $SN$ at line~\ref{line:mwmrma:read_sn} and also the sequence number read from $R$ in some previous \rd{} operation ($sn = prev\_sn$, line~\ref{line:mwmrma:no_new_write}). Hence, as in the case of a silent \wrt{}, $R.seq = x$ before $SN$ is read in $op$.

  By Lemma~\ref{lem:R}, there exists a unique value $v_x$ such that while $R.seq = x$, we have $R.val = v_x$. To change $(R.seq,R.val)$ to $(x,v_x)$, a successful \cas{} is applied to $R$ at line~\ref{line:mwmrma:writer_cas} by some process $p'$ while performing a \wrt{}($v_x$)  operation $op'$. $op'$ is visible and by definition $sn(op') = x$. For uniqueness, suppose that there is another visible \wrt{} operation $op''$ with $sn(op'') = x$. The \cas{} applied to $R$ by this operation has arguments of the form $(\_,lsn,\_),(\_,x,\_)$ with $lsn < x$ (otherwise, the repeat loop terminate with the break statement at line~\ref{line:mwmrma:writer_break}). But once $R.seq$ is changed to $x$, such a \cas{} cannot succeed as sequence numbers stored in $R.seq$ are increasing (Invariant~\ref{obs:R_S_increase 2}). 
\end{proof}

\begin{lemma}
  \label{lem:rd_wrt_terminate}
  If an operation (\rd{}, \wrt{} or \adt{}) $op$ terminates in $H'$,
  then $\mathit{SN} \geq sn(op)$. 
\end{lemma}
\begin{proof}
  If $op$ is a direct \rd{}, a  \wrt{} or an \adt{}, this follows from the \cas{} applied to $SN$ before the operation returns (line~\ref{line:mwmrma:reader_cas},   line~\ref{line:mwmrma:writer_return} or line~\ref{line:mwmrma:audit_return}) that tries to change the value of $SN$ from $sn-1$ to $sn$. In each case,  the value of the local variable $sn$ is the sequence number $sn(op)$ associated with $op$. Moreover, when this \cas{} is applied,  $SN \geq sn(op) -1$. Indeed, if $op$ is a \wrt{}, $\mathit{SN} = sn(op) -1$ when it is read at the beginning of $op$ (line~\ref{line:mwmrma:writer_read_sn}). In the other cases, $sn(op)$ is fetched or read from $R.seq$, therefore, by Lemma~\ref{lem:SN_R}, $SN \geq sn(op)-1$ immediately after this step.
  Since $SN$ is increasing (Invariant~\ref{obs:R_S_increase 1}), $SN \geq sn(op) -1$ when $SN.\cas(sn(op)-1,sn(op))$ is applied, and hence $SN \geq sn(op)$ after this step whether or not the \cas{} fails.  

  If $op$ is a silent \rd{}, $sn(op)$ is the value read from $\mathit{SN}$ at the beginning of the operation (line~\ref{line:mwmrma:read_sn}. Since the values stored in $\mathit{SN}$ are increasing, $\mathit{SN} \geq sn(op)$ when $op$  terminates. 
\end{proof}

\begin{restatable}{lemma}{RTOrder}
  \label{lem:rt_order}
  If the response of an operation $op$ precedes the invocation of $op'$ in $H'$,
  then $op$ precedes $op'$ in $L(\alpha)$. 
\end{restatable}

\begin{proof}
  Assume, towards a contradiction, 
  that $op$ completes before the invocation of $op'$ in $\alpha$, 
  but $op'$ is placed before $op$ in $L(\alpha)$. 
  We examine several cases, 
  according to the linearization rules used to place $op$ and $op'$ in $L(\alpha)$:
  \begin{itemize}
  \item Both $op$ and $op'$ are linearized using rule $R1$. $op$ and $op'$ are ordered in $\alpha$ following the order in which a step in their execution interval occur in $\alpha$. It is thus not possible that $op'$ is placed before $op$ in $L(\alpha)$.
    
  \item $op$ is linearized using rule $R1$, and $op'$ using rule $R2$ or $R3$. Let $x = sn(op)$ be the sequence number of $op$, and $ls$, its linearization step. This step either changes $R.seq$ to $x$ (step $\rho_x$) if $op$ is a visible \wrt{} or $R.seq = x$ when it is applied (if $op$ is a silent \rd{} linearized with rule $R1$, $\emph{SN} = R.seq =x$ when $\emph{SN}$ is read.). As $ls$ occurs in the execution interval of $op$ and $R.seq$ is increasing, $R.seq \geq x$ and,  by Lemma~\ref{lem:rd_wrt_terminate} $SN \geq x$ when $op$ terminates.

    As $op'$ starts after $op$ terminates, and as both $\mathit{SN}$ and $R.seq$ are increasing, we still have $R.seq \geq x$ and  $SN \geq x$ when $op'$ starts. Hence, $op'$ reads $x' \geq x$ from $SN$ (line~\ref{line:mwmrma:read_sn} or line~\ref{line:mwmrma:writer_read_sn}), and, following rules $R2/R3$ is linearized immediately before $\rho_{x'+1}$ (which changes $R.seq$ to $x'+1 > x$). It thus appears in $L(\alpha)$ after every operation with sequence number $x$ linearized with rule $R1$.

  \item $op$ is linearized using rule $R2$ or $R3$, and $op'$ using rule $R1$. Let $x$ be the value read from $SN$ in $op$ (line~\ref{line:mwmrma:read_sn} or line~\ref{line:mwmrma:writer_read_sn}). As $op$ is not placed using rule $R1$, $R.seq \geq x+1$ when $op$ terminates. Indeed, if $op$ is silent \rd{} $R.seq \geq x+1$ when $SN$ is read. Otherwise, $op$ is a silent \wrt{}, and thus $R.seq$ has already been updated to a value $\geq x+1$ when a $\mathsf{read}$ (line~\ref{line:mwmrma:writer_read_R}) or  a $\cas{}$ (line~\ref{line:mwmrma:writer_cas}) to $R$ is applied in $op$. Therefore, the linearization step of $op'$ is applied to a configuration in which $R.seq \geq x+1$, and thus occurs after $\rho_{x+1}$. Hence $op$ is placed in $L(\alpha)$ after the visible \wrt{} with sequence number $x+1$, whereas $op'$ is placed before by definition of rules $R2$/$R3$, which is a contradiction. 

  \item Rule $R1$ is not used to linearize $op$ and $op'$. Let $x$ and $x'$ be the values of $\mathit{SN}$ read at the beginning of $op$ and $op'$ respectively. As $op$ precedes $op'$ in $\alpha$, $x \leq x'$. If $x < x'$, $op'$ is placed after the visible \wrt{} with sequence number $x+1$, and $op$ before this \wrt{} in $L(\alpha)$. If $x = x'$, we remark that $op$ cannot be a \wrt{} operation. Indeed, if $op$ is a \wrt{}, $sn(op) = x+1$, and therefore by Lemma~\ref{lem:rd_wrt_terminate}, $SN \geq x+1$ when $op$ terminates and hence also when $op'$ starts. Hence, $op$ and $op'$ are placed using the same rule or $op$ is placed using rule R2 and  $op'$, rule $R3$. In the latter case, $op$ is placed before $op'$ by rule $R3$. In the former case,  $op$ cannot appear after $op'$ in $L(\alpha)$ as they are relatively ordered with respect to the order in which a step taken in their execution interval occurs in $\alpha$. \qedhere
  \end{itemize}
\end{proof}

\begin{lemma}
  \label{lem:read_valid}
  If a \rd{} operation $rop$ in $H'$ returns $v$,
  then $v$ is the input of the last \wrt{} that precedes $rop$ in $L(\alpha)$. 
\end{lemma}

\begin{proof}
  Let $x = sn(rop)$. Suppose that $rop$ is direct, then $v = v_x$, the value fetched from $R$ at line~\ref{line:mwmrma:read_fetch} and we have $x = R.seq$ when this \fx{} is applied. $(R.seq,R.val)$ is changed to $(x,v_x)$ by a \cas{} (line~\ref{line:mwmrma:writer_cas}) applied  in a \wrt{} operation with input $v_x$, and by Lemma~\ref{lem:R}, this step is unique. Let $wop$ be the operation that applies this \cas{}. $wop$ is a visible \wrt{}, linearized according to rule $R1$ before $op$, with the step $\rho_{x}$ (the \cas{} that changes $(R.seq,R.val)$ to $(x,v_x)$). Note that there  is no  visible  \wrt{} operation placed  between $wop$ and $rop$ in $L(\alpha)$ (otherwise $rop$ will not read $x$ from $R.seq$ at line~\ref{line:mwmrma:read_fetch}), and thus every silent \wrt{} is placed (according to rule R3) before $wop$  or after $rop$. $rop$ thus returns  the input of the last \wrt{} that precedes it in $L(\alpha)$.

  Otherwise, suppose that $rop$ is silent. 
  By the code, $rop$ is preceded by a direct \rd{} operation $dop$ performed by the same process,
  which returns the same value, and with the same sequence number $x$. 
  Let $wop$ be the last $\wrt{}$  operation that precedes $dop$ in $L(\alpha)$. 
  As shown above, the input of $wop$ is $v_x = v$ and $wop$ is the unique visible \wrt{}  with sequence number $x$. 
  If no visible \wrt{} is placed between $wop$ and $rop$, 
  then there is also no silent \wrt{}  between $wop$ and $rop$ in $L(\alpha)$ 
  (as rule $R3$ places  a silent \wrt{}  immediately before a visible \wrt{}). 
  Assume, towards a contradiction, that there is a visible \wrt{}($v_{x'}$) between $wop$ and $rop$ in $L(\alpha)$, 
  with $seq(wop') = x'$. By Lemma~\ref{lem:R}, $x' > x$. 
  If $rop$ is placed using rule $R1$, $\rho_{x'}$ occurs before $SN$ is read by $rop$ (line~\ref{line:mwmrma:read_sn}). As $rop$ is placed using  rule $R1$, we thus have $R.seq = SN \geq  x'$  when this $\mathsf{read}$ is applied since $R.seq$ is non-decreasing. Therefore, $sn(rop)\neq x$, which is a contradiction. $rop$ is thus placed using rule $R2$. It is thus before the visible \wrt{} with sequence number $x+1$, and  hence there is no visible \wrt{} between $wop$ (which the visible \wrt{} with sequence number $x$) and $rop$. 
\end{proof}

\begin{restatable}{lemma}{AuditComp}
\label{lem:audit_completness}
If an \rd{} operation $rop$ by process $p_j$ returns $v$ and appears before 
an \adt{} operation $aop$ in $L(\alpha)$,
then $(j,v)$ is contained in the response set of $aop$. 
\end{restatable}


\begin{proof}
If $rop$ is silent, then it is preceded by a direct \rd{} $rop'$ by the same process,
which returns the same value. In that case, we consider $rop'$ instead of $rop$.
So, assume $rop$ is direct. 
Let $x_r$ and $x_a$ denote  respectively $sn(rop)$ and $sn(aop)$. 
Since both $aop$ and $rop$ are linearized by rule $R1$, 
$x_r \leq x_a$ since $rop$ precedes $aop$ in $L(\alpha)$. 

If $x_r = x_a = x$, the \fx{} applied by $rop$ is before $R$ is $\mathsf{read}$ 
in $aop$ at line~\ref{line:mwmrma:audit_read}. 
As $x_r= x_a = x$, this read step returns a triple $(x,rv,rb)$ 
where $rb[j] \neq rand_x[j]$ and  $rv = v$. 
Therefore, $(j,v)$ is included in the audit set $A$ 
(line~\ref{line:mwmrma:audit_non_definitive}).

If $x_r < x_a$, consider the visible \wrt{} operation in which $R.seq$ 
is changed from $x_r$ to $x_r+1$ (step $\rho_{x_r+1}$). 
Before  applying this step, 
a writer $p$ sets $B[x_r][j]$ to true and $V[x_r]$ to $v$ (line~\ref{line:mwmrma:writer_record_reads}). 
Indeed, if $R.bits$ is modified by $p_j$ after $p$ reads $R$ at line~\ref{line:mwmrma:writer_read_R}, 
the \cas{} at line~\ref{line:mwmrma:writer_cas} trying to change $R.seq$ to $x_r+1$ fails. 
Note also that by Lemma~\ref{lem:R} no other value $v' \neq v$ is written to $B[x_r]$. 
By the code, $B[x_r][j]$ is read in an \adt{} operation only after $R.seq$ is seen to be larger than or equal $x_r + 1$ at line~\ref{line:mwmrma:audit_read}. 
Hence, $B[x_r][j]$ and $V[s]$ are read by $aop$ or by a preceding \adt{} of the same process 
after $\rho_{x_r+1}$. 
It thus follows that $(j,v)$ is added to the audit set $A$ before $aop$ returns.
\end{proof}

\begin{restatable}{lemma}{AuditAccur}
\label{lem:audit_accuracy}
If a pair $(j,v)$ is contained in the response set of an \adt{} operation $aop$, 
then there is a \rd{} operation by process $p_j$ that returns $v$ 
and appears before $aop$ in $L(\alpha)$. 
\end{restatable}


\begin{proof}
Let $x = sn(aop)$.  
One way for the pair $(j,v)$ to be included in the response set $A$ of $aop$, 
is if $j$ is extracted from the bit-string read from $R$ at line~\ref{line:mwmrma:audit_non_definitive}.  
Let $(rs,rv,rb)$ be the triple read from $R$ at line~\ref{line:mwmrma:audit_read}. 
Note that $x = sn(op) = rs$.  
Hence, $R$ is previously changed to $(x,rv=v_x,rand_{x})$ (in step $\rho_x$).
Since $rb[j] \neq rand_{x}[j]$, a \fx{} by $p_j$ is applied to $R$ 
after $\rho_x$ and before $R$ is read by $aop$. 
This $\fx$ is applied during a \rd{} by $p_j$ that returns $v_x = v$. 
This \rd{} is direct, and like $aop$, is linearized by rule $R1$,
implying it precedes $aop$ in $L(\alpha)$. 

Otherwise, $aop$ reads $true$ from a Boolean register $B[s][j]$
(line~\ref{line:mwmrma:audit_definitive}), for some $s < x$. 
Before $B[s][j]$ is read, a \wrt{} operation by some process $p_i$ sets $B[s][j]$ to true (line~\ref{line:mwmrma:writer_record_reads}). 
By the code, $p_i$ has previously read  a triple $(s,v',rb)$ from $R$ 
( where $rb[j] \neq rand_s[j]$.  
Therefore, as above, $p_j$ applies a \fx{} to $R$ when $R.seq = s$ in 
a \rd{} operation $rop$. 
This operation is a direct \rd{} that returns $v'$ 
(since by Lemma~\ref{lem:R}, $R.seq = s' \implies R.val =v'$), 
and its place in $L(\alpha)$ is determined by its linearization step $ls$, 
which is the \fx{} applied to $R$.
On the other hand, the linearization step of $aop$ is the $\mathsf{read}$ of $R$ 
(line~\ref{line:mwmrma:audit_read}), and $s < x = R.seq$, when this step occurs. 
Therefore, $R$ is \textsf{read} in $aop$ after $\rho_{s+1}$, 
that is after $R.seq$ is changed from $s$ to $s+1$. 
Before this step is applied,  $B[j][s]$ is set to true (line~\ref{line:mwmrma:writer_record_reads}), 
and hence the direct \rd{} $rop$ is linearized before $aop$. 
\end{proof}

\section{Correctness Proof for Algorithm~\ref{alg:maxreg} (Auditable Max Register)}
\label{sec:maxregister_proof}

The proof is divided into four parts. First we check that executions of the max register algorithm have a simple structure, as for the register implementation. Each execution $\alpha$ may be partitioned into phases, in which sequence numbers in registers $\mathit{SN}$ and $R$ are \emph{E}qual or \emph{D}iffer by one. Phases are associated with unique increasing values, which are the only values that can be returned by \rd{} operations. 
The second part then shows, essentially along the lines of the proof of wait-freedom of Algorithm~\ref{alg:cas_mwa}, that \rd{}, \wrtm{} and \adt{} operations are wait-free. Relying on the structural lemmas of the first part, we prove in part three that each execution $\alpha$ is linearizable. The basis of the linearization $L(\alpha)$ is a linearization $L(\beta)$ of an  execution $\beta$ of Algorithm~\ref{alg:cas_mwa} indistinguishable from $\alpha$ for any reader or auditor. The construction allows to lift the strong auditing properties of Algorithm~\ref{alg:cas_mwa} to the max register implementation. 
The fourth, and last, part of the proof establishes that \rd{} and \wrtm{} operations are uncompromised by the readers. 
Until the last part, pairs $(values,nonce)$ are considered as single opaque values, ordered lexicographically.

\paragraph{Partition into phases}
Recall that  $R.seq$, $R.val$  denote  respectively the sequence number and the value stored in $R$. We observe that the sequence numbers in $\mathit{SN}$ and $R.seq$ follow the same pattern as in  Algorithm~\ref{alg:cas_mwa}, namely the successive values of $(R.seq,\mathit{SN})$ are $(0,0),(1,0),(1,1),(2,1),\ldots$

Indeed, when $\mathit{SN}$ is changed, it is incremented by one (line~\ref{lmr:writer_read2_SN}, line~\ref{lmr:writer_return}) and whenever $R.seq$ is changed to $x+1$, $x$ has previously been read from $\mathit{SN}$ (at line~\ref{lmr:writer_read1_SN} or line~\ref{lmr:writer_read2_SN}). In fact, in Algorithm~\ref{alg:maxreg}, each iteration of the repeat loop behaves as a \wrt{} instance of Algorithm~\ref{alg:cas_mwa}. A sequence number $x$ is first read from $\mathit{SN}$  (in line~\ref{lmr:writer_read1_SN} for the first iteration, line~\ref{lmr:writer_read2_SN} otherwise), then if $R.seq < x+1$ (line~\ref{lmr:writer_read_R}),  an attempt to changes $R.seq$ to $x+1$ (together with $R.val$ and $R.bits$) is made by applying a \cas{} (line~\ref{lmr:writer_cas_R}) before, if successful, making sure that $\mathit{SN} \geq x+1$ (line~\ref{lmr:writer_return}). Lemma~\ref{lem:epoch} thus still holds. It is restated below for convenience:

\begin{lemma}
\label{lem:epoch_wrtm}
A finite execution $\alpha$ can be written either as 
$E_0 \rho_1 D_1 \sigma_1 E_1$ $\ldots \rho_k D_k \sigma_k E_k$ or as 
$E_0 \rho_1 D_1 \sigma_1 E_1 \ldots$  $\sigma_{k-1} E_{k-1} \rho_k D_k$, 
for some integer $k \geq 0$, where:
  \begin{itemize}
  \item $\rho_\ell$ and $\sigma_\ell$ are the steps that respectively 
  change the value of $R.seq$ and $\mathit{SN}$  from $\ell-1$ to $\ell$ 
  ($\rho_\ell$ is a successful \cas{} line~\ref{lmr:writer_cas_R}, 
  $\sigma_\ell$ is also a successful \cas{}, applied within a \rd{}, line~\ref{line:mwmrma:reader_cas}, 
  a \wrt{}, line~\ref{lmr:writer_sn_used} or line~\ref{lmr:writer_return},  or an \adt{}, line~\ref{line:mwmrma:audit_return}).
  \item in any configuration in $E_\ell$, $R.seq = \mathit{SN} = \ell$,  and in any configuration in $D_\ell$, $R.seq = \ell = \mathit{SN} +1$. 
  \end{itemize}
\end{lemma}

Whenever $R.val$ is  changed to $v$, $v$ has previously been read from the max register $M$ (line~\ref{lmr:writer_read_M}).
Therefore, is easy to see that:

\begin{invariant}
\label{obs:R_val_increase}
The successive values of $R.val$ are strictly increasing. 
\end{invariant}

During two consecutive phases $D_xE_x$, neither the sequence number nor the  value stored in $R$  change, 
and as we have just seen, $R.val$ is increasing, at most one unit ahead of $\mathit{SN}$. Therefore, similarly to Lemma~\ref{lem:R} for Algorithm~\ref{alg:cas_mwa}, we have

\begin{lemma}
  \label{lem:SN_R_val_increase}
  Let $\alpha$ be a finite execution. There exists $k\geq 0$ and 
  $v_1< \ldots < v_k$ such that the sequence of values of the  fields $(R.seq,R.val)$ is $(0,v_0),(1,v_1),\ldots,(k,v_k)$. 
\end{lemma}

\paragraph{Wait-freedom}
The code of \rd{} and \adt{} operations is the same in Algorithm~\ref{alg:cas_mwa} and in Algorithm~\ref{alg:maxreg}, 
so they are wait-free as shown in Appendix~\ref{app:proof of register}.
For a \wrtm{} operation $op$, as for \wrt{},  concurrent \rd{} operations may prevent $op$ from successfully applying a \cas{} to $R$ and hence from exiting the repeat loop. This happens at most $m$ times, where $m$ is the number of readers, as implied by Lemma~\ref{lem:reader_once} which still holds.  Unlike for \wrt{} operations, the repeat loop continues (skipping the remainder  of the current iteration) if the current sequence number $sn$ has already been associated with a value ($R.seq \geq sn$, line~\ref{lmr:writer_sn_used}). However, we show  that this can happen a constant number of times before $R.val$ becomes greater than the input of $op$.

\begin{lemma}[wait-freedom of \wrtm{}]
  \label{lem:wrtm_wait_free}
  Every \wrtm{} operation terminates within a finite number of its own steps. 
\end{lemma}
\begin{proof}
  Let $wop$ be a \wrtm{} operation by some process $p$ with input $w$, and assume towards a contradiction, that it does not terminate in some infinite execution $\alpha$. 

  We claim that  after $w$ is written to $M$ in line~\ref{lmr:writer_write_M}, $(R.seq, R.val)$ changes at most once. To see why, let  $(\ell,v_{\ell})$ be the value of $(R.seq,$ $R.val)$ immediately after $p$ writes $w$ to $M$. By Lemma~\ref{lem:SN_R_val_increase}, if $(R.seq,$ $R.val)$ is next changed at least twice, its two subsequent values are $(\ell+1,v_{\ell+1})$ and  $(\ell+2,v_{\ell+2})$ with $v_{\ell} < v_{\ell+1}  < v_{\ell+2}$. Let $q$ be the process that  changes $(R.seq,R.val)$ from
  $(\ell+1,v_{\ell+1})$ to $(\ell+2,v_{\ell+2})$ by applying  a successful \cas{} in line~\ref{lmr:writer_cas_R}. Before applying this \cas{}, $q$ in that order reads $(\ell+1,v_{\ell+1})$ from $R$ (line~\ref{lmr:writer_read_R}) and $v_{\ell+2}$ from $M$ (line~\ref{lmr:writer_read_M}). Each of these steps occur after $w$ is written to $M$ by $p$. Because $M$ is a max register, $v_{\ell+2} \geq w$, and therefore, $R.val \geq w$ after $(R.seq,R.val)$ has been changed to $(\ell+2,v_{\ell+2})$. 
  Since $wop$ does not terminate, and $p$ reads $R$ (line~\ref{lmr:writer_read_R}) in each iteration of repeat loop, it eventually discovers that $R.val \geq w$, and exits  the loop with the break statement  (line~\ref{lmr:writer_break}): a contradiction.

  Let therefore $(\ell',v_{\ell'})$ be the final value of $(R.seq,R.val)$. After $R.seq$ is set to $\ell'$, by Lemma~\ref{lem:epoch_wrtm}, $\mathit{SN} \in \{\ell'-1,\ell'\}$. In the first iteration  in which $p$ reads $(\ell',v_{\ell'})$ from $R$ (line~\ref{lmr:writer_read_R}), if $\ell' \geq x+1$, where $x$ is the last value read from $\mathit{SN}$ by $p$ before this iteration,  $p$ reads therefore $\ell'-1$ or $\ell'$ from $\mathit{SN}$ (line~\ref{lmr:writer_sn_used}). If $\ell'-1$ is read, then in the following iteration, if $\mathit{SN}$ has not meanwhile been changed to $\ell'$, the \cas{} applied to $\mathit{SN}$ succeeds, and $p$ finally reads $\ell'$ from $\mathit{SN}$. To summarize, there is a configuration $C$ in $\alpha$ after which the following always holds (1) $R.seq = \ell' = \mathit{SN}, R.val = v_{\ell'} < w$ and (2) for process $p$, $sn = \ell'+1$.

  The rest of the proof is the same as  the proof of wait-freedom for \wrt{} operation in Algorithm~\ref{alg:cas_mwa}. Consider an iteration of the repeat loop that starts after $C$, and let $(sr_1,vr_1,br_1)$ be the triple read from $R$ in this iteration. Note that $sr_1 = \ell' < sn = \ell'+1$ and $vr_1 = v_{\ell'} < w$. Therefore, $p$ applies a \cas{} to $R$ at the end of this iteration (line~\ref{lmr:writer_cas_R}), which fails since $wop$ does not terminate. Let $(sr_2,vr_2,br_2)$ be the value of $R$ immediately before this \cas{} is applied. 
  Since $R.seq$ and $R.val$  no longer change, $br_2 \neq br_1$: at least one reader applies a \fx{} to $R$ during this iteration of the repeat loop.
  The same reasoning applies to the  next iterations. In each of them, $R.seq$ and $R.val$ are the same, and thus a reader applies a \fx{} before the \cas{} of line~\ref{lmr:writer_cas_R}. By Lemma~\ref{lem:reader_once}, each reader applies at most once \fx{} to $R$ while it holds the same sequence number: a contradiction.
\end{proof}

\paragraph{Linearizability} Let $\alpha$ be a finite execution, and $H$ be the history of the \rd{}, \wrtm{} and \adt{} operations in $\alpha$. 

We  define an execution $\beta$ of Algorithm~\ref{alg:cas_mwa} that is indistinguishable from $\alpha$ for any reader and any auditor. This is made possible by the fact that \rd{} and \adt{} share the same code in both Algorithm~\ref{alg:cas_mwa} and Algorithm~\ref{alg:maxreg}. To linearize $\alpha$, we start from $L(\beta)$ (which contains every terminated \rd{} and \adt{} of $\alpha$), replace each \wrt{} with a \wrtm{} with the same input, and then place the remaining terminated \wrtm{} operations. These last operations are \emph{silent}, since their input is never read. 

The construction of execution $\beta$ is as follows. By   Lemma~\ref{lem:epoch_wrtm} and Lemma~\ref{lem:SN_R_val_increase}, there exists values  $v_1 < \ldots < v_k$ such that $\alpha$ can be written as $\alpha = E_0\rho_1D_1\sigma_1E_1\ldots\rho_kD_k\sigma_kE_k$ or $\alpha = E_0\rho_1D_1\sigma_1E_1\ldots\rho_kD_k$ $\sigma_kE_kD_k$. 
Let $q_1,\ldots,q_k$ be the (not necessarily distinct) processes that apply steps $\rho_1,\ldots,\rho_k$, respectively. Recall that $\rho_x$ changes $(R.seq,R.val)$ from $(x-1,v_{x-1})$ to $(x,v_x)$.  Before applying $\rho_x$, process $q_x$ reads $x$ from $\mathit{SN}$ (in line~\ref{lmr:writer_read1_SN} or line~\ref{lmr:writer_read2_SN}), reads a triple $(x-1,v_{x-1},b_{x-1})$ from  $R$ (in line~\ref{lmr:writer_read_R}), writes $v_{x-1}$ to $V[x-1]$ and depending on $b_{x-1}$, appropriately sets to \emph{true} some registers in the array $B[x-1]$ (line~\ref{lmr:writer_record_reads}). This sequence of steps is denoted $A_x$. A key observation is that   $A_x\rho_x$ is the sequence of steps applied by $q_x$ in a visible \wrt{} operation with input $v_x$ in some execution of Algorithm~\ref{alg:cas_mwa}. 

$\beta$ is the execution obtained by removing from $\alpha$ every step by \wrtm{} operations, except, for each $x, 1 \leq x \leq k$, steps $A_x,\rho_x$ and $\sigma_x$. Indeed, removed steps are failed attempts to modify $R$ or $\mathit{SN}$ or are reads  and writes to  $M$, which is not accessed by reader and auditor. They are therefore invisible for readers and auditors. A removed step may be also a $\mathsf{write}(v_x)$ to $V[x]$  or setting some register $B[x][j]$ to true. This is indiscernible for the auditors, since $V[x]$ and the $B[x]$ are set to their final value by $q_x$ when applying $A_x$, and no auditors access $V[x]$ and $B[x]$ before $\rho_x$. 

We then  remove   all invocations and responses of \wrtm{} operations and, instead, we place for each $x, 1 \leq x \leq k$, an invocation of $\wrt{}(v_x)$ by $q_x$ immediately before $A_x$, and a matching response (except perhaps for $x=k$) immediately after $\sigma_x$. Finally, in $\beta$, each step $\sigma_x$ is applied by $q_x$. We obtain:  

\begin{claim}
  \label{claim:valid_and_indistinguishable}
  $\beta$ is a valid execution of Algorithm~\ref{alg:cas_mwa} and, if $p$ is a reader or an auditor, then $\indp{\alpha}{\beta}$.  
\end{claim}

There are exactly $k$ \wrt{} operations in $\beta$: $\wrt{}(v_1), \ldots, \wrt{}(v_k)$. For each $x, 1 \leq x \leq k$, $\wrt{}(v_x)$  is classified as visible, since it applies a successful \cas{} to $R$, and $sn(\wrt{}(v_x)) = x$. As shown in Section~\ref{subsec:register_proof}, the linearization  $L(\beta)$ of $\beta$   includes in particular the operations $\wrt{}(v_1), \ldots, \wrt{}(v_k)$ in that order. 

A \wrtm{} operation $op$ with input $w$ is classified as  \emph{visible} if there exists $x, 1 \leq x \leq k$ such that  $w = v_x$ and step $\rho_x$ is in the execution interval of $op$. 
Otherwise, if $op$ terminates, it is classified as \emph{silent}. 
Note that for each $x, 1 \leq x \leq k$, a visible \wrtm{} exists,
since a $\wrtm{}(v_x)$ is invoked before $R.val$ is changed to $v_x$. 
This \wrtm{} operation cannot terminate before $R.val \geq v_x$ 
or before applying the  \cas{} $\rho_x$ that changes $R.val$ to $v_x$.

The next two technical lemmas will be used for showing that $L(\alpha)$ extends the real-time order between operations.

\begin{lemma}
  \label{lem:sigma_visible_wrtm}
  If operation \wrtm{}($v_x$) is visible, then $\sigma_x$ is in the execution interval of $op$. 
\end{lemma}
\begin{proof}
$\sigma_x$ is the successful \cas{} that changes $\mathit{SN}$ from $x-1$ to $x$.     By definition, $\rho_x$ is in the execution interval of $op$. 
Since $\sigma_x$ follows $\rho_x$ in $\alpha$, the lemma is true if $op$ does not terminate.

  If $op$ terminates, then since its input is $v_x$, it reads a value $v \geq v_x$ from $R$ or applies a successful \cas{} that changes $R.val$ to $v \geq v_x$ (line~\ref{lmr:writer_cas_R}). 
  Since the successive values of $R.val$ are $v_0 < v_1< \ldots < v_k$, $v \in \{v_x,\ldots,v_k\}$. If $v \in \{v_{x+1},\ldots,v_k\}$, it follows from Lemma~\ref{lem:epoch_wrtm} that $\emph{SN} \geq x$ when $R$ is read or the \cas{} of line~\ref{lmr:writer_cas_R} applied. $\sigma_x$  therefore occurs before this step.   
  Since $\rho_x$  is in the execution interval of $op$ and precedes $\sigma_x$, $\sigma_x$ is also in the execution interval of $op$.

  Suppose now that  $v = v_x$. If  the repeat loop terminates after $R$ is read (break statement of line~\ref{lmr:writer_break}),  $R.seq \geq x$ when this read is applied (Lemma~\ref{lem:epoch_wrtm}) and hence $sn \geq x$ at the end of the loop. Similarly, if the loop terminates after applying a successful \cas{} that changes $R.val$ to $v_x$, this step by Lemma~\ref{lem:epoch_wrtm} also changes $R.seq$ to $x$ and therefore $sn = x$ at the end of the loop. In both cases, the \cas{} in line~\ref{lmr:writer_return} tries to changes $\mathit{SN}$ from $x-1$ to $x$. If it succeeds, $\sigma_x$ is in the execution interval occurs. If not, $\sigma_x$ has already occurred, and, since it follows step $\rho_x$ which is in the execution interval of $op$,  $\sigma_x$ is also in this interval.
\end{proof}

\begin{lemma}
  \label{lem:rho_silent_wrtm}
  Let $op$ be a silent \wrtm{}  with input $w$ satisfying $v_{x-1} < w \leq v_x$, for some $x, 1 \leq x \leq k$. The last step of $op$ follows $\rho_x$ in $\alpha$. 
\end{lemma}
\begin{proof}
   $\rho_x$ is the successful \cas{} that changes $R.val$ from $v_{x-1}$ to $v_x$. Let $p$ be the process that performs $op$. 
  Since $op$ terminates, there exists $v \geq w$ such that $p$ reads $v$ from $R.val$ or successfully  applies a \cas{} that  changes $R.val$ to $v$.  In both   cases, since the successive values of $R.val$ are $v_0 < v_1,<  \ldots < v_k$, $v \geq v_x$. Therefore, in the first case $\rho_x$ precedes the read of $R$ in $op$.  In the second case, the successful \cas{} is either $\rho_x$ or $\rho_{x'}$ for some $x' \geq x$. 
  This is not the last step in $op$, since $p$ tries to update $\mathit{SN}$ before returning (line~\ref{lmr:writer_return}). 
\end{proof}

Let $H'$ be the complete history obtained by completing or removing non-terminated operations in $H$ as follows: \rd{} and \adt{} operations that do not appear in $L(\beta)$ are removed. 
These operations do not terminate in $\beta$, and thus, also in $\alpha$,
since $\indp{\alpha}{\beta}$ for any auditor or reader $p$. 
For every non-terminated \rd{} or \adt{} operation $op$ that appears in $L(\beta)$, 
we add a response for $op$, as in the sequential execution $L(\beta)$, at the end of $H$. 
For every $x, 1 \leq x \leq k$, and non-terminated visible $\wrtm(v_x)$, 
we add a response at the end of $H$, in arbitrary order. 
Every unclassified \wrtm{} operation is removed.

To define $L(\alpha)$, we start (1) from $L(\beta)$.   We then (2) replace for each $x, 1 \leq x \leq k$,  $\wrt{}(v_x)$ by the  set of visible  \wrtm{}($v_x$) operations, ordered arbitrarily. Finally, we  (3) place each remaining silent \wrtm{}($w$) operation $op$, respecting the real time precedence  between $op$ and already placed operations, and after \wrtm{}($v_x$), where $v_{x-1} < w \leq v_x$. 

For each $x, 1 \leq x \leq k$,  \wrt{}($v_x$) is a visible \wrt{} with sequence number $x$. It is thus linearized in $L(\beta)$ with $\rho_x$ (Rule R1, Section~\ref{subsec:register_proof}). Since $\rho_x$ is in the execution interval of  visible \wrtm{}($v_x$) operations, step (2) does not break the real-time precedence with operations linearized in $L(\beta)$ with one of their steps. 

For step (3), 
silent \wrtm{} operations are first sorted by their real time order. 
They are then placed one after the other as follows. 
Operation \wrtm{}($w$), with $v_{x-1} < w \leq v_x$ is placed after $\rho_x$, 
that is, after visible operations \wrtm{}($v_x$), 
and immediately after every already-placed operation that precedes it in $\alpha$. 
This always possible, since by Lemma~\ref{lem:rho_silent_wrtm}, $op$ does not terminate before $\rho_x$. 

\begin{lemma}
  \label{lem:wrtm_rt_order}
  Let $op$, $op'$ be two operations in $H'$. If $op$ ends before $op'$ starts in $\alpha$, $op$ precedes $op'$ in $L(\alpha)$. 
\end{lemma}
\begin{proof}
  The lemma is true if $op$ or $op'$ is a silent \wrtm{}, since each silent \wrtm{} is placed in $L(\alpha)$ before every operation it precedes, and after every operation it follows in the real-time order. 
  The lemma is also true if $op$ and $op'$ are \rd{} or \adt{} operations. Indeed, $op$ ends before  $op'$ starts also in $\beta$, and therefore appears before $op'$ in $L(\beta)$, and thus also in $L(\alpha)$. We examine the remaining cases next:
  \begin{itemize}
  \item $op$ and $op'$ are two visible \wrtm{}. Let $v_{x}$ and $v_{x'}$ be their respective inputs. 
  By definition, $\rho_x$ and $\rho_{x'}$ are in the execution interval of $op$ and $op'$,  respectively.
  Therefore, $x < x'$ and $v_x < v_{x'}$, since $op$ ends before $op'$ starts. 
  In $L(\beta)$, $\wrt{}(v_x)$ and $\wrt{}(v_{x'})$ are linearized according to the order in which their associated linearization steps occur in $\beta$ (rule R1). These steps are $\rho_x$ and $\rho_{x'}$. 
  Therefore,  $\wrt{}(v_x)$ is before  $\wrt{}(v_{x'})$ in $L(\beta)$, and hence by step (2) of the construction of $L(\alpha)$,  $\wrtm{}(v_x)$ precedes also $\wrtm{}(v_{x'})$ in $L(\alpha)$.

  \item $op$ is a \rd{} or an \adt{} and $op$ is a visible \wrtm{}. Since $op$ is linearized in $L(\beta)$, it has a sequence number $x = sn(op)$, which is the value read from $\mathit{SN}$ (line~\ref{line:mwmrma:read_sn}) if $op$ is a silent \rd{}, fetched or read from $R$ (line~\ref{line:mwmrma:read_fetch}) if $op$ is a direct \rd{} or an \adt{}.  Let $v_{x'}$ be the input of $op$.   
  Since $op$ ends before $op'$ starts, and by definition of visible \wrtm{}, $\rho_{x'}$ is in the execution interval of $op'$, $\mathit{SN}$ is read or $R$ fetched/read  before $\rho_{x'}$. Hence $sn(op) < x'$, and it thus follows that $op$ is placed before the \wrt{} operation with input $v_{x'}$ in $L(\beta)$. Therefore, by step (2) of the construction of $L(\alpha)$, $op$ precedes $op'$ in $L(\alpha)$. 
    
  \item $op$ is a visible \wrtm{} and $op'$ is a \rd{} or an \adt{}. As in the previous case, let $v_x$ be the input of $op$, and let $x' = sn(op)$. If $op'$ has a linearization step $ls(op')$, that is $op'$ is an \adt{}, a direct \rd{} or a silent \rd{} in which $\mathit{SN}.\mathsf{read}$ is applied during $E_{x'}$, $ls(op')$ follows $\rho_x$ in $\alpha$ and thus also in $\beta$. Hence $op'$ appears after $\wrt{}(v_x)$ in $L(\beta)$. Therefore, by step (2) of the construction, $op'$ is after $op$ in $L(\alpha)$.

    It remains to examine the case in which $op'$ is a silent \rd{} without a linearization step. This means that $\emph{SN}$ is read in $op$ (in line~\ref{line:mwmrma:read_sn}) during phase $D_{x'+1}$. 
    Since $\rho_x$ and, by Lemma~\ref{lem:sigma_visible_wrtm}, also $\sigma_x$ are included in $op$'s execution interval, phase $D_x$ is contained in $op$'s execution interval. Since $op$ ends before $op'$ starts, we therefore have $x < x'+1$.  $op'$ is placed in $L(\beta{})$ according to rule $R2$ immediately before the \wrt{} operation \wrt{}($v_{x'+1}$). 
    Since $x < x'+1$, $op'$ is placed after \wrt{}($v_{x}$) in $L(\beta)$, and thus after $op= \wrtm{}(v_x)$ in $L(\alpha)$. \qedhere
  \end{itemize}
\end{proof}

We next prove that $L(\alpha)$ is a sequential execution of an auditable max register.

\begin{lemma}
  \label{lem:wrtm_read_return_max}
  If a \rd{} operation $rop$ by some process $p$ in $H'$ returns $v$, then $v$ is the largest input of a \wrtm{} that precedes $rop$ in $L(\alpha)$.
\end{lemma}
\begin{proof}
  Let $v_x$ be the  input of the last visible \wrtm{} that precedes $rop$ in $L(\alpha)$. We claim that $v_x$ is  the largest input of the (silent or visible) \wrtm{} that precedes $rop$ in $L(\alpha)$. Let $w$ be the input of a silent \wrtm{} operation $wop$ that precedes $op$ in $L(\alpha)$. By step (3) of the construction of the linearization, $wop$ is preceded in $L(\alpha)$ by a visible $\wrtm{}$ operation with input $v_{x'} \geq w$. Therefore, $v_x \geq w$ and the claim follows.

  Since $\wrtm{}(v_x)$ is the last \wrtm{} operation that precedes $rop$ in $L(\alpha)$, $\wrt{}(v_x)$ is the last \wrt{} operation that precedes $rop$ in $L(\beta)$. 
  Since $L(\beta)$ is a linearization of an execution $\beta$ of  register implementation (Algorithm~\ref{alg:cas_mwa}), $rop$ returns $v_x$ in execution $\beta$. Therefore, since $\indp{\alpha}{\beta}$, $rop$ returns $v_x$  in execution $\alpha$. 
\end{proof}

\begin{lemma}
  \label{lem:wrtm_audit_return}
  A pair $(j,v)$ is contained in the response set of an \adt{} operation $aop$ if and only if there is a \rd{}  operation by process $p_j$ that returns $v$ and appears before $aop$ in $L(\alpha)$. 
\end{lemma}
\begin{proof}
  Let $op$ be a \rd{} operation by process $p_j$ that returns $v$ and precedes $aop$ in $L(\alpha)$. Let $q$ be the process that invokes $aop$. By construction, $op$ precedes $aop$ also in $L(\beta)$. Since $\ind{\alpha}{p_j}{\beta}$, $op$ returns $v$ in $\beta$, and, as seen in the proof of Algorithm~\ref{alg:cas_mwa} (Lemma~\ref{lem:audit_completness}), the response set of $aop$ in $\beta$ contains $(j,v)$. 
  Since $\ind{\alpha}{q}{\beta}$, the response set of $aop$ contains $(j,v)$ also in $\alpha$.

  Reciprocally, suppose that $(j,v)$ is included in the response set of an \adt{} operation $aop$ by some process  $q$. 
  Since $\ind{\alpha}{q}{\beta}$, $aop$ reports $(j,v)$ also in execution $\beta$, 
  and therefore, there exists a \rd{} operation $rop$ by $p_j$ that precedes $aop$ in $L(\beta)$ and returns $v$ (Lemmal~\ref{lem:audit_accuracy}). By construction, $rop$ also precedes $aop$ in $L(\alpha)$, and since $\ind{\beta}{p_j}{\alpha}$,  returns $v$ in $\alpha$. 
\end{proof}

\paragraph{Auditabilty and uncompromised operation instances}
The characterization (recalled below) of effective \rd{} operations, established in Section~\ref{subsec:register_proof} for Algorithm~\ref{alg:cas_mwa}  holds, as the proof can be easily adapted.


\begin{claim}
  \label{claim:eff_charac}
  A \rd{} operation $rop$ by $p_j$  is $v$-effective in $\alpha$ if and only if it has returned $v$ or it is pending and either  (1) $p_j$ has read $x$ from $\mathit{SN}$, $x = prev\_sn$ (line~\ref{line:mwmrma:read_sn})  and $prev\_val = v$ 
or (2) $p_j$ has applied  \fx{} to $R$ (line~\ref{line:mwmrma:read_fetch}), from which it reads $v$ from $R.val$. 
\end{claim}

Essentially, audit properties are lifted from the auditable register implementation, thanks to the construction of $L(\alpha)$ from a linearization $L(\beta)$ of an execution of that algorithm. 

\begin{lemma}
  \label{lem:rdm_effective_lin}
  A \rd{} operation $rop$ that is invoked in $\alpha$ is in $L(\alpha)$ if and only if $rop$ is effective in $\alpha$. 
\end{lemma}
\begin{proof}
  Suppose that \rd{} operation  $rop$ by $p_j$ is effective in $\alpha$. $rop$ is also effective in $\beta$ since $\ind{\alpha}{p_j}{\beta}$ and being effective is a local property. Indeed, it follows from the characterization (Claim~\ref{claim:eff_charac}) that to determine if a given \rd{} operation by some process $q$ is effective, it is enough to examine the steps of $q$.  The same  lemma holds for the register implementation (Lemma~\ref{lem:eff_read_lin}), and hence $rop$ is in $L(\beta)$. 
  Since $L(\alpha)$ extends $L(\beta)$, $rop$ is in $L(\alpha)$ as well.

  Conversely, suppose that $rop$ is in $L(\alpha)$. 
  By construction, it is also in $L(\beta)$ and hence $rop$ is effective in $\beta$ by Lemma~\ref{lem:eff_read_lin}. 
  Since $\ind{\beta}{p_j}{\alpha}$, as explained above, $rop$ is effective in $\alpha$.
\end{proof}

As in the proof of Algorithm~\ref{alg:cas_mwa}, Lemma~\ref{lem:wrtm_audit_return} and Lemma~\ref{lem:rdm_effective_lin} imply:

\begin{lemma}
  \label{lem:wrt:audit_effective}
If an \adt{} operation $aop$  is invoked and returns in an extension $\alpha'$ of $\alpha$, and $\alpha$ contains a $v$-effective \rd{} operation by process $p_j$, then $(j,v)$  is contained in the response set of $aop$.  
\end{lemma}

So far, we have ignored the nonce $N$ adjoined to input $w$ of  \wrtm{} operations, treating $(w,N)$ as a single opaque value. 
We now use them to prove that a reader cannot compromise a \wrtm{}($v$), unless it actually reads $v$. 
\begin{lemma}[uncompromised \wrtm]
    \label{lem:wrtm_non_auditable}
  For every value $w$, and every reader $p_j$ either there is a \rd{} operation by $p_j$ in $\alpha$ that is $w$-effective, or there exists $\alpha'$, $\ind{\alpha'}{p_j}{\alpha}$ in which no \wrtm{} has input $w$. 
\end{lemma}
\begin{proof}
  Suppose that $p_j$ has no $w$-effective \rd{} in $\alpha$. If there is no \wrtm{} operation with input $w$, taking $\alpha = \alpha'$ proves the lemma.

  Assume that $w$ is the input of a \wrtm{} operation $op$ in $\alpha$. Let $u$  be the largest input of \wrtm{}  in $\alpha$ smaller than $w$, and let $N$ be the nonce associated to it.  Execution $\alpha'$ is the same as execution $\alpha$, except that the input of $op$ is $u$, and the nonce is $N'$ where $N < N'$. Note that $(u,N) <  (u,N') < (w,N)$, since pairs $(value,nonce)$ are ordered lexicographically. Also, for any other pair $v = (Val,M)$ in $\alpha$, $(Val,M) < (u,N')$ or $(w,N) < (Val,M)$. Therefore, any comparison  between $(u,N')$ and another pair $v$ has the same outcome as a comparison between $(w,N)$ and $v$. Since in Algorithm~\ref{alg:maxreg} the pairs (value,nonce) are only tested for equality or compared, the sequence $(0,v_0), \ldots (k,v_k$) of (sequence number, pair) successively stored  in $(R.seq,R.val)$ is the same in $\alpha$ and $\alpha'$, except if in $\alpha$, $v_x = (w,N)$ for some $x, 1 \leq x \leq k$. In that case, $v_x = (u,N')$ in $\alpha'$.

  If $(w,N)$ is never written to $R$, neither is $(u,N')$ and therefore $\ind{\alpha}{p_j}{\alpha'}$. If $(w,N)$ is written to $R$, $p_j$ does not apply a \fx{} to  $R$ while $R.val = (w,N)$, since otherwise the corresponding \rd{} is $w$-effective. Therefore, $p_j$ does not apply a \fx{} to $R$ in $\alpha'$ while $R.val = (u,N')$ and hence $\ind{\alpha}{p_j}{\alpha'}$. If $\alpha'$ has no \wrtm{} with input $w$, this proves the lemma. Otherwise,  the same construction, applied to $\alpha'$ leads to an execution $\alpha''$, $\alpha\stackrel{p_j}{\sim}\alpha'\stackrel{p_j}{\sim}\alpha''$ which has one less \wrtm{} operation with input $w$. This can be repeated until every \wrtm{}($w$) has been eliminated. 
\end{proof}

\begin{lemma}[uncompromised \rd{}]
    \label{lem:rdm_non_auditable}
    Let  $p_j \neq p_k$ be two readers. There is an execution $\ind{\alpha'}{p_j}{\alpha}$ in which no \rd{} by $p_k$ is $v$-effective.
\end{lemma}
\begin{proof}
  The same lemma holds for the register implementation (Lemma~\ref{lem:rd_non_auditable}).
  Hence, there exists an execution $\beta', \ind{\beta'}{p_j}{\beta}$ of Algorithm~\ref{alg:maxreg} in which no \rd{} by $p_k$ is $v$-effective. 
  Since $\ind{\beta}{p_j}{\alpha}$, we have that $\ind{\beta'}{p_j}{\alpha}$, implying the lemma. 
\end{proof}

We conclude:
\begin{theorem}
  \label{thm:maxreg}
  Algorithm~\ref{alg:maxreg} is a wait-free, linearizable implementation of an auditable, multi-writer max register. Moreover, in any execution $\alpha$, an \adt{} reports $(j,v)$ if and only if $p_j$ has a $v$-effective \rd{} in $\alpha$, each \wrt{}($v$) is uncompromised by a reader $p_j$ unless it has a $v$-effective \rd{} and, each \rd{} by $p_k$ is uncompromised by a reader $p_j \neq p_k$.   
\end{theorem}
\end{document}